\documentclass[11pt]{article}
\usepackage[utf8]{inputenc}
\usepackage{indentfirst}
\usepackage{bm}
\usepackage{amsthm}
\usepackage{amsmath}
\usepackage{amssymb}
\usepackage{graphicx}
\usepackage{float}
\usepackage{color}
\usepackage{hyperref}
\usepackage{verbatim}
\usepackage{subfig}

\allowdisplaybreaks

\newcommand{\reff}[1]{{\rm (\ref{#1})}}

\newcommand{\R}{\mathbb{R}}            % real numbers
\newcommand{\x}{\mbox{\boldmath$x$}}

\newtheorem{theorem}{Theorem}[section]
\newtheorem{lemma}[theorem]{Lemma}

\newtheorem{remark}[theorem]{Remark}

\def\XXint#1#2#3{{\setbox0=\hbox{$#1{#2#3}{\int}$}
\vcenter{\hbox{$#2#3$}}\kern-.51\wd0}}

\newcommand{\dd}{{ \rm d}}

\definecolor{Light Turquoise}{RGB}{84,135,120} % comment color used by Xu

\title{Mean-field Modeling of Social Interactions Using Classical Density Functional Theory
 }
\author{Ziheng Xu\thanks{School of Mathematical Sciences, MOE-LSC, and CMA-Shanghai, Shanghai Jiao Tong University, Shanghai, 200240, China. Email: cauchy\_schwarz@sjtu.edu.cn}
\and	
Shenggao Zhou\thanks{School of Mathematical Sciences, MOE-LSC, and CMA-Shanghai, Shanghai Jiao Tong University, Shanghai, 200240, China. Email: sgzhou@sjtu.edu.cn}}
\date{\today}

\begin{document}

\maketitle

\begin{abstract}
Incorporating social interactions is essential to an accurate modeling of epidemic spreading. This work proposes a novel local mean-field density functional theory model by using the sum-of-exponential approximation of convolution kernels for social interactions, which in turn converts
the convolution terms into interaction potentials that are governed by the Debye-H\"uckel equation. Thanks to the local formulation of the
proposed model, linear stability analysis is able to derive a novel instability condition associated with cross interactions. Global existence of the solution to the proposed model with a simplified self-repulsive interaction potential is established. Extensive numerical simulations are performed to 
assess the impact of cross social interactions on transmission and isolation, verify the instability conditions obtained from linear stability analysis, and provide theoretical guides for the control of disease spreading. 
% exists widely in human society and it is a key factor in epidemic spreading. From the mathematical view, social distancing could be evaluated by mean-field approximation. 
% % derive the SIRDH model
% In this paper, we propose a compartmental model, combining susceptible-infected-recovered (SIR) model, dynamical density functional theory (DDFT) and mass-action transmission mechanism together. In the process of constructing the model, we use mean-field approximation and investigate the effect of social distancing on epidemic spreading. 
% % linear stability analysis
% To study the homogeneous steady state, we perform the linear stability analysis and give the instability criteria. 
% % pde analysis
% Moreover, we give the global existence theorem under proper conditions. 
% % numerical results
% Additionally, in numerical simulation, we discover some patterns at intermediate stages.  

\end{abstract}

% \tableofcontents
% \newpage

\section{Introduction} \label{Sec. introduction}
Mathematical modeling, particularly through differential equations, plays a pivotal role in modern epidemiology~\cite{murray_book02_vol1,murray_book02_vol2, Keeling_PUP08}. Mathematical models allow us to predict outbreak trajectories and assess control strategies before costly implementation. During the COVID-19 pandemic, such models enabled governments to project medical shortages and evaluate lockdown efficacy. Classical mathematical models of disease transmission predominantly adopt reaction–diffusion frameworks that describe population density through the Fickian diffusion and mass-action kinetics~\cite{WangCaiWu_NonlinearAnal12, GaiIron_JMB20, Banerjee_PhysD24}. Social interventions, such as the practice of social distancing and isolation of infected individuals, are often neglected in the reaction–diffusion modeling. The resulting models often overestimate spatial homogeneity and underestimate the impact of individual behaviors under social restrictions~\cite{tevrugt_NatCommun20}. 

To include individual effects, microscopic agent-based models that treat each person as a stochastic particle moving on dynamic contact networks are able to faithfully capture heterogeneity and social interactions. However, microscopic models are often computation demanding and provide little analytical insight, hindering their applications in policy formulation. Coarse-grained continuum theories based on statistical description of interacting particles are often more efficient than agent-based simulations, while still faithfully capturing individual interactions~\cite{Rosenfeld_PRL89, MarconiTarazona_JCP99, ArcherEvans_JCP04, Bickmann_JPCM20, tevrugt_AdvPhys20}. Dynamic models for interacting agents have been derived for various applications, such as the cancer growth~\cite{ArcherWard_PRE18}, active matter~\cite{Wensink_PRE08, Wittkowski_MP11, Wittmann_EPL16}, ecology~\cite{Lopez_PRL13}, and socio-economic systems~\cite{CarrilloVecil_Book, Calzola_PhysD24}, showing great promise in modeling epidemiology. 

The classical density functional theory (DFT) is an effective coarse-graining tool to derive macroscopic models with connection to microscopic dynamics~\cite{MarconiTarazona_JCP99, ArcherEvans_JCP04, tevrugt_AdvPhys20}. It begins with the microscopic Smoluchowski equation for many-particle probability density and integrates out all degrees of freedom except the one-body density field. The adiabatic approximation that two-body correlations relax much faster than one-body density yields a deterministic evolution equation. Recently, the dynamical density functional theory has been utilized to propose a coarse-grained model in a seminal work~\cite{tevrugt_NatCommun20}, to describe the epidemic spreading including the effects of social distancing and isolation. In addition to the Brownian diffusion process, the social interactions among interacting agents are described by a mean-field approximation through nonlocal convolution of the population density with an interaction kernel. The infectious reaction kinetics is modeled by the well-known susceptible-infected-recovered (SIR) model~\cite{Keeling_PUP08, GaiIron_JMB20, Banerjee_CSF23}. Coupling of such ingredients yields a time-dependent integro-differential system. Linear stability analysis performed on such a system derives an instability condition related to a well-known outbreak criterion ascribing to infectious reaction, whereas it cannot directly assess the impact of social interactions on instability due to its convolution formulation. 

Along the line of the mean-field approximation, this work proposes a local mean-field model for epidemic spreading with social interactions, based on the nonlocal dynamic density functional theory developed in~\cite{tevrugt_NatCommun20}. It starts with the sum-of-exponential (SOE) approximation of the interaction kernel of convolutions for social interactions. Such an approximation in turn converts the convolution terms into interaction potentials that are governed by the Debye-H\"uckel equation with sources arising from the corresponding density field. This results in a local DFT model (resembling classical chemotaxis models) with SIR reaction kinetics. Linear stability analysis on the proposed local model elucidates that the homogeneous steady state solution becomes unstable if the transmission rate is larger than the well-known outbreak criterion, which has been pointed out in the work~\cite{tevrugt_NatCommun20}. Furthermore, the linear stability analysis also derives a new instability condition associated with cross interactions, thanks to the local formulation of the proposed model. 

To further understand the newly derived model, we also study the global existence of the solution to the model with a simplified self-repulsive potential. Following the parabolic theory on chemotaxis-type systems, we first establish the $L^p$ estimate on the local solution and then the $L^\infty$ estimate following the technique of the Moser-Alikakos iteration. The global existence is finally proved using a relevant extensibility criterion. Extensive numerical simulations are performed to investigate the impact of social interactions on epidemic spreading. Numerical simulations with various combinations of interaction intensity present a phase diagram on instability of the homogeneous steady state solution, further verifying that the instability conditions derived in the linear stability analysis can robustly predict instability. In addition, numerical studies are conducted to assess the impact of cross social interactions on transmission, isolation, etc., providing theoretical guides for the control of disease spreading. 

This paper is structured as follows. In Section \ref{s:model}, we derive the SIR Debye-H\"uckel model based on the dynamical density functional theory. In Section \ref{Sec. Linear stability analysis}, we perform linear stability analysis on a homogeneous base state. In Section \ref{Sec. PDE analysis}, we present global existence of the solution to a simplified version of the model. 
% with  with proper parameter settings. 
In Section \ref{Sec. Numerical results}, we report numerical results. Finally, in Section \ref{Sec. conclusion}, we draw some conclusions.

\section{Model} \label{s:model}
\subsection{Density Functional Theory with SIR Reaction}\label{ss:DFT-SIR}
We model persons in epidemic spreading as interacting particles based on the time-dependent classical density functional theory~\cite{MarconiTarazona_JCP99, ArcherEvans_JCP04, tevrugt_NatCommun20}, which is well-known to be able to accurately describe the diffusive relaxation of interacting particles under assumptions that the motion of particles is Markovian and ergodic.  The theory has been successfully used to describe various applications, such as   active matter~\cite{Wensink_PRE08, Wittkowski_MP11, Wittmann_EPL16},
cancer growth~\cite{ArcherWard_PRE18}, ecology~\cite{Lopez_PRL13}, and collective behavior in socio-economic sciences~\cite{CarrilloVecil_Book, Calzola_PhysD24}.

Consider an interacting particle system that consists of $M$ species. Denote the density of each species of particles as  $\rho_j$ $(j=1,2,\cdots,M)$. Let $\rho=(\rho_1, \rho_2, \cdots, \rho_M)$. The free energy of the particle system is given by~\cite{tevrugt_NatCommun20} 
\begin{equation}\label{F}
\mathcal{F}[\rho] = \mathcal{F}_{\rm id}[\rho] + \mathcal{F}_{\rm exc}[\rho] + \mathcal{F}_{\rm ext}[\rho],
\end{equation}
where the first  term accounts for the entropy of particles described as ideal gas without interactions:
\[
\mathcal{F}_{\rm id} [\rho]= \beta^{-1} \sum_{j=1}^M \int \rho_j \left( \log(\Lambda^3 \rho_j)-1 \right) \dd \x.
\]
Here $\beta$ is the inverse thermal energy and $\Lambda$ is thermal de Broglie wavelength. The second term is the excess free energy to include long-range and short-range particle interactions. For instance,  the fundamental measure theory is proposed to express the excess Helmholtz free energy for hard-sphere fluids~\cite{Rosenfeld_PRL89}. Later, the fundamental measure theory is further extended to describe multicomponent hard-sphere mixture using the Mansoori--Carnahan--Starling--Leland equation of state~\cite{WhiteBear_JPCM02, YuWu_JCP02}. In the context of social interactions, it is more appropriate to use soft interaction potentials~\cite{tevrugt_NatCommun20}, rather than the hard-sphere potential, such as the Gaussian core model~\cite{LBH_PRE00, ArcherEvans_PRE01} that can be well approximated by the mean-field description.  With such description,  the excess free energy is given by 
\[
\mathcal{F}_{\rm exc} [\rho]= \frac{1}{2} \sum_{i=1}^M \sum_{j=1}^M \iint w_{ij}  \rho_i(x) U(|\x-\x'|) \rho_j(x')\, \dd \x \dd \x', 
\]
where $U: \R^+ \to \R$ is a kernel function to describe distance-dependent social interactions and $W: =(w_{ij})_{M \times M}$ is a symmetric matrix with \emph{positive} elements $w_{ij}$ representing the interaction intensity between individuals of species $i$ and $j$. The third term in \reff{F} is the external free energy given by
\[
\mathcal{F}_{\rm ext} = \sum_{j=1}^M \int \rho_j V_j^{\rm ext} \dd \x,
\]
where the given external potential $V_j^{\rm ext}$ is independent of $\rho$.  

By the continuity equation, one obtains a dynamic evolution equation for each species:
\[
\partial_t \rho_j = \Gamma_j \nabla \cdot (\rho_j \nabla \mu_j), 
\]
where $\Gamma_j$ is the mobility for the $j$th species and $\mu_j$ is the chemical potential given by
\[
\mu_j := \frac{\delta \mathcal{F}}{\delta \rho_j} = \beta^{-1} \log(\Lambda^3 \rho_j) + \sum_{l=1}^M w_{jl} (U \ast \rho_l)+V_j^{\rm ext}.
\]
The dynamics corresponds to the $H^{-1}$-gradient flow of the free energy $\mathcal{F}$.  In summary, the evolution of each species described by the time-dependent classical density functional theory is governed by
\begin{equation} \label{DDFT}
    \partial_t \rho_j = D_j \Delta \rho_j + \Gamma_j \nabla \cdot \left(\rho_j \sum_{l=1}^M w_{jl} \nabla(U \ast \rho_l)+\rho_j \nabla V_j^{\rm ext}\right),
\end{equation}
where  $D_j := \beta^{-1} \Gamma_j$ is the diffusion coefficient.

To describe epidemic spreading, we now combine the derived dynamics \reff{DDFT} for the description of population diffusion and interactions with the well-known SIR reaction model for disease transmission~\cite{tevrugt_NatCommun20}:
\begin{equation} \label{SIR-DDFT}
    \left \{
    \begin{split}
        \partial_t S &= D_S \Delta S + \Gamma_S \nabla \cdot  \left[ S \nabla \left( U \ast( w_{SS} S+ w_{SI} I + w_{SR} R )\right)+S \nabla V_S^{\rm ext} \right]-\lambda SI, \\
        \partial_t I &= D_I \Delta I + \Gamma_I \nabla \cdot  \left[ I \nabla \left( U \ast( w_{IS} S+ w_{II} I + w_{IR} R )\right)+I \nabla V_I^{\rm ext} \right] +\lambda SI -\gamma I, \\
        \partial_t R &= D_R \Delta R + \Gamma_R \nabla \cdot  \left[ R \nabla \left( U \ast( w_{RS} S+ w_{RI} I + w_{RR} R )\right) +R \nabla V_R^{\rm ext}\right] 
        + \gamma I,
    \end{split}
    \right.
\end{equation}
where $S(\x,t)$, $I(\x,t)$, and $R(\x,t)$ are densities of the susceptible, infectious, and recovered individuals, $\lambda$ is transmission rate, and $\gamma$ is recovery rate.
\subsection{Sum-of-exponential Approximation}\label{ss:DH}
The system~\reff{SIR-DDFT} includes diffusion, reaction, and social interactions based on the mean-field approximation involving convolution. Such an integro-differential system is numerically intractable, when stable, implicit discretization is taken into account.  Furthermore, as remarked in the work~\cite{tevrugt_NatCommun20}, linear stability analysis directly performed on such a system can only reveal the well-known outbreak criterion ascribing to reaction, whereas it cannot directly assess the impact of social interactions on instability. Following the spirit of the mean-field approximation, we propose to further approximate the social interaction kernel $U(\cdot)$ and derive a local approximation model that can be tackled numerically and analytically.
 
In the context of social interactions, it is reasonable to assume that in order to practice social distancing, the social interaction kernel $U(\cdot)$ is repulsive,  
singular at the origin, and fast decaying as the distance goes to infinity. We take the kernel $U(r)= r^{-(\alpha+1)}$ with $\alpha>0$ as an example. The following sum-of-exponential (SOE) approximation is available for $f(r) := rU(r)$.
\begin{theorem}\cite{beylkin1, beylkin2} \label{Thm SOE approx}
Let $f(r)= r^{-\alpha}$. For any $\alpha>0$ and given accuracy $\epsilon>0$, there exist $\{\kappa_l\}_{l=1}^\infty$ and positive $\{c_l\}_{l=1}^\infty$ such that
\begin{equation}\label{SOE-Error}
\left |f(r) - \sum_{l=1}^\infty c_l \frac{e^{-\kappa_l r}}{4 \pi} \right| \leq f(r) \epsilon, \mbox{ for all }  r>0.
\end{equation}
\end{theorem}
\noindent According to \reff{SOE-Error}, we have
\[
\left |\frac{U(r)-\sum_{l=1}^\infty c_l e^{-\kappa_l r}/(4 \pi r)}{U(r)}  \right| \leq \epsilon, \mbox{ for all }  r>0,
\]
which suggests that the following finite sum can approximate $U(r)$ with a small relative error~\cite{beylkin2}:
\begin{equation}\label{SOE-U}
U(r) \approx U_N(r) = \sum_{l=1}^N c_l \frac{e^{-\kappa_l r}}{4 \pi r},
\end{equation}
where $N$ is the number of the finite sum. Various model reduction methods have been proposed to effectively reduce $N$ in literature~\cite{xu_JSC22, beylkin1, beylkin2}.

With the approximation~\reff{SOE-U}, we introduce a potential $\psi_S$ induced by $S$:
\[\psi_S := U_N \ast S= \sum_{l=1}^N c_l \phi_l^S,
\]
where $\phi_l^S:=\frac{e^{-\kappa_l r}}{4\pi r} \ast S $ can be shown to solve the Debye-H\"uckel equation in $\R^3$:
\begin{equation} \label{DH eq}
    -\Delta \phi_l^S +  \kappa_l^2  \phi_l^S = S.
\end{equation}
Analogously, we can introduce $\psi_I$, $\psi_R$, $\phi_l^I$, and $\phi_l^R$. 

With the Debye-H\"uckel potential, the mean-field integro-differential system \reff{SIR-DDFT} can be approximated by 
\begin{equation} \label{DH-SIR-DFT}
    \left\{
        \begin{split}
            &\partial_t S = D_S \Delta S + \Gamma_S \nabla \cdot \left[S\nabla \left( \sum_{l=1}^N c_l \left(w_{SS} \phi_l^S + w_{SI} \phi_l^I + w_{SR} \phi_l^R\right)\right)+S \nabla V_S^{\rm ext}\right] -\lambda SI, \\
            &\partial_t I = D_I \Delta I + \Gamma_I \nabla \cdot \left[I\nabla \left( \sum_{l=1}^N c_l \left(w_{IS} \phi_l^S + w_{II} \phi_l^I + w_{IR} \phi_l^R\right)\right)+I \nabla V_I^{\rm ext}\right] +\lambda SI -\gamma I,  \\
            &\partial_t R = D_R \Delta R + \Gamma_R \nabla \cdot \left[R\nabla \left( \sum_{l=1}^N c_l \left(w_{RS} \phi_l^S + w_{RI} \phi_l^I + w_{RR} \phi_l^R\right)\right)+R \nabla V_R^{\rm ext}\right] +\gamma I, \\
            & -\Delta \phi_l^S +  \kappa_l^2 \phi_l^S = S,\enspace l = 1,\cdots,N, \\
            & -\Delta \phi_l^I + \kappa_l^2 \phi_l^I = I,\enspace l = 1,\cdots,N, \\
            & -\Delta \phi_l^R + \kappa_l^2 \phi_l^R = R,\enspace l = 1,\cdots,N, \\
        \end{split}
    \right.
\end{equation}
where $w_{ij}$ for $i,j =S, I, R$ is the positive interaction intensity, and $\{\kappa_l\}_{l=1}^N$ and $\{c_l\}_{l=1}^N$ are treated as parameters. 
To further study the system~\reff{DH-SIR-DFT} analytically and numerically, we consider the system on a bounded domain $\Omega$ and impose the following initial conditions and homogeneous Neumann boundary conditions
\begin{equation}\label{Init+BC}
    \left\{
        \begin{split}
            &S(\bm{x},0) = S_0(\bm{x}), I(\bm{x},0) = I_0(\bm{x}), R(\bm{x},0) = R_0(\bm{x}), \\
            &\nabla S \cdot \bm{n} = \nabla I \cdot \bm{n} = \nabla R \cdot \bm{n} = \nabla \phi^S_l \cdot \bm{n} = \nabla \phi^I_l \cdot \bm{n} = \nabla \phi^R_l \cdot \bm{n} = 0 \enspace \text{on} \enspace \partial \Omega, \\
        \end{split}
    \right.
\end{equation}
where $\bm{n}$ is a unit exterior normal vector on the boundary $\partial \Omega$, and $S_0(\bm{x})$, $I_0(\bm{x})$, and $R_0(\bm{x})$ are initial data. From now on, we abbreviate the model \reff{DH-SIR-DFT}, which couples the density functional theory with the SIR reaction using the Debye-H\"uckel potential, as {\bf DFT-SIR-DH}.

\section{Linear stability analysis} \label{Sec. Linear stability analysis}
We now perform linear stability analysis on the system \reff{DH-SIR-DFT} to assess the impact of both reaction and social interactions on the development of solution instability. For simplicity, the external potentials are set as zero. We consider a homogeneous steady state of \reff{DH-SIR-DFT} that is given by
\begin{equation} \label{SIR DH model s.s. eq of N approx}
    \left\{
        \begin{split}
            &\lambda S_{\rm hom}^* I_{\rm hom}^* = 0, &\\
            &I_{\rm hom}^* = 0, &\\
            &\kappa_l^2 \phi_{l,\rm hom}^{S,*} = S_{\rm hom}^*,& \enspace l = 1,\cdots,N, \\
            &\kappa_l^2 \phi_{l,\rm hom}^{I,*} = I_{\rm hom}^*,& \enspace l = 1,\cdots,N, \\
            &\kappa_l^2 \phi_{l,\rm hom}^{R,*} = R_{\rm hom}^*,& \enspace l = 1,\cdots,N. 
        \end{split}
    \right.
\end{equation}
% From $(\ref{SIR DH model s.s. eq of N approx})$, it is obvious that $I_{\rm hom}^* = 0$. 
% By linear stability analysis, we have the following theorem concerning the instability.
Linear stability analysis on the steady state \reff{SIR DH model s.s. eq of N approx} reveals the following conditions on instability due to reaction and social interactions.
\begin{theorem}[Linear stability analysis] \label{t:LBA}
The steady state solution \reff{SIR DH model s.s. eq of N approx} to the DFT-SIR-DH model is linearly unstable if one of the following conditions holds 
\begin{subequations}
    \begin{align}
        \label{SIR DH lin.stb. condition Reaction}    &\lambda S_{\rm hom}^* > \gamma, \\
        \label{SIR DH lin.stb. condition Energy QB}   &D_S + D_R + \left( \Gamma_S S_{\rm hom}^* w_{SS} + \Gamma_R R_{\rm hom}^* w_{RR} \right) \sum\limits_{l=1}^N \frac{c_l}{\kappa_l^2} < 0, \\
        \label{SIR DH lin.stb. condition Energy QC}   &w_{SR}w_{RS} > \left(w_{SS}+\frac{1}{\beta S_{\rm hom}^* {\sum}_{l=1}^N c_l/\kappa_l^2}\right)\left(w_{RR}+\frac{1}{\beta R_{\rm hom}^*  {\sum}_{l=1}^N c_l/\kappa_l^2 }\right).  
    \end{align}
\end{subequations}
\end{theorem}
\begin{proof}
Consider perturbation to the base state:
\begin{equation} \nonumber
\begin{split}
    (S,I,R,\phi_l^S,\phi_l^I,\phi_l^R)^T =& (S_{\rm hom}^*,I_{\rm hom}^*,R_{\rm hom}^*,\phi_{l,hom}^{S,*},\phi_{l,hom}^{I,*},\phi_{l,hom}^{R,*})^T \\
    &+ (\delta S,\delta I,\delta R,\delta \phi_l^S,\delta \phi_l^I,\delta \phi_l^R)^T, \enspace l = 1,\cdots,N, 
\end{split}
\end{equation}
where $\delta S$, $\delta I$, $\delta R$, $\delta \phi_l^S$, $\delta \phi_l^I$, and $\delta \phi_l^R$ are infinitesimal perturbation. Substituting the perturbed solution into the DFT-SIR-DH model \reff{DH-SIR-DFT} and ignoring higher order terms, one obtains
\begin{equation} \label{SIR DH model lin.stb. eq}
    \left\{
        \begin{split}
           & \partial_t \delta S = D_S \Delta \delta S + \Gamma_S S_{\rm hom}^* \Delta \delta \Phi_S -\lambda S_{\rm hom}^* \delta I, \\
           &\partial_t \delta I = D_I \Delta \delta I + (\lambda S_{\rm hom}^* - \gamma) \delta I, \\
           &\partial_t \delta R = D_R \Delta \delta R + \Gamma_R R_{\rm hom}^* \Delta \delta \Phi_R + \gamma \delta I,\\
           &-\Delta \delta \phi_l^S + \kappa_l^2 \delta \phi_l^S = \delta S, \enspace l = 1,\cdots,N,  \\
           &-\Delta \delta \phi_l^I + \kappa_l^2 \delta \phi_l^I = \delta I, \enspace l = 1,\cdots,N,  \\
           &-\Delta \delta \phi_l^R + \kappa_l^2 \delta \phi_l^R = \delta R, \enspace l = 1,\cdots,N, 
        \end{split}
    \right.
\end{equation}
where
\begin{equation}\nonumber
    \begin{split}
        \delta \Phi_S &:= \sum\limits_{l=1}^N c_l (w_{SS}\delta \phi_l^S + w_{SI}\delta \phi_l^I + w_{SR}\delta \phi_l^R) ,  \\
        \delta \Phi_R &:= \sum\limits_{l=1}^N c_l (w_{RS}\delta \phi_l^S + w_{RI}\delta \phi_l^I + w_{RR}\delta \phi_l^R).
    \end{split}
\end{equation}
Consider
\[
(\delta S,\delta I,\delta R,\delta \phi_l^S,\delta \phi_l^I,\delta \phi_l^R)^T = e^{\sigma t}\eta_k(\bm{x}) (\Tilde{S},\Tilde{I},\Tilde{R},\Tilde{\phi_l^S},\Tilde{\phi_l^I},\Tilde{\phi_l^R})^T, \enspace l = 1,\cdots,N,
\]
where $\{\eta_k\}_{k \ge 0}$ are eigenfunctions of the Laplacian given by
\begin{equation}\nonumber
    \left\{
        \begin{split}
            &-\Delta \eta_k = k^2 \eta_k \enspace \text{in} \enspace \Omega,  \\
            &\nabla \eta_k \cdot  \bm{n} = 0 \enspace \text{on} \enspace \partial \Omega, 
        \end{split}
    \right.
\end{equation}
with the $k$ being the magnitude of the Fourier-space vector. By direct calculations, 
\begin{equation} \label{SIR DH model lin.stb. fourier eq 6 var.}
    \left\{
        \begin{split}
            &\sigma \Tilde{S} = -k^2 (D_S \Tilde{S} + \Gamma_S S_{\rm hom}^* \Tilde{\Phi}_S) -\lambda S_{\rm hom}^* \Tilde{I}, \\
            &\sigma \Tilde{I} = -k^2 D_I \Tilde{I} +(\lambda S_{\rm hom}^* - \gamma)  \Tilde{I}, \\
            &\sigma \Tilde{R} = -k^2 (D_R \Tilde{R} + \Gamma_R R_{\rm hom}^* \Tilde{\Phi}_R) +\gamma \Tilde{I}, \\
            &(k^2 + \kappa_l^2) \Tilde{\phi_l^S} = \Tilde{S}, \enspace l = 1,\cdots,N, \\
            &(k^2 + \kappa_l^2) \Tilde{\phi_l^I} = \Tilde{I}, \enspace l = 1,\cdots,N, \\
            &(k^2 + \kappa_l^2) \Tilde{\phi_l^R} = \Tilde{R}, \enspace l = 1,\cdots,N,
        \end{split}
    \right.
\end{equation}
where
\begin{equation}\nonumber
    \begin{split}
        \Tilde{\Phi}_S &:= \sum\limits_{l=1}^N c_l (w_{SS} \Tilde{\phi_l^S}  + w_{SI}\Tilde{\phi_l^I} + w_{SR}\Tilde{\phi_l^R}) ,  \\
        \Tilde{\Phi}_R &:= \sum\limits_{l=1}^N c_l (w_{RS}\Tilde{\phi_l^S} + w_{RI}\Tilde{\phi_l^I} + w_{RR}\Tilde{\phi_l^R}).
    \end{split}
\end{equation}
Eliminating $\Tilde{\phi_l^S},\Tilde{\phi_l^I},\Tilde{\phi_l^R}$ in (\ref{SIR DH model lin.stb. fourier eq 6 var.}), one arrives at
% \begin{equation} \label{SIR DH model lin.stb. fourier eq 3 var.}
\[
    \left\{
        \begin{split}
            &\left[\sigma + k^2 (D_S + \Gamma_S S_{\rm hom}^* w_{SS} Y_k)\right]\Tilde{S} + \left(\lambda S_{\rm hom}^* + k^2 \Gamma_S S_{\rm hom}^* w_{SI} Y_k\right)\Tilde{I} + k^2\Gamma_S S_{\rm hom}^* w_{SR} Y_k\Tilde{R} = 0, \\
            &\left[\sigma + k^2 D_I - (\lambda S_{\rm hom}^* - \gamma)\right]\Tilde{I} = 0, \\
            &k^2\Gamma_R R_{\rm hom}^* w_{RS} Y_k\Tilde{S} + (-\gamma + k^2 \Gamma_R R_{\rm hom}^* w_{RI} Y_k)\Tilde{I} + \left[\sigma + k^2 (D_R + \Gamma_R R_{\rm hom}^* w_{RR} Y_k)\right]\Tilde{R}= 0,
        \end{split}
    \right.
\]
% \end{equation}
where
\[ Y_k := \sum\limits_{l=1}^N \frac{c_l}{k^2 + \kappa_l^2}.  \]
% Since $(\ref{SIR DH model lin.stb. fourier eq 3 var.})$ gives rise to a linear equation for nonzero $\Tilde{S}, \Tilde{I}, \Tilde{R}$, the determinant of coefficient matrix in $(\ref{SIR DH model lin.stb. fourier eq 3 var.})$ equals zero, \emph{i.e.}
In order to get a nonzero solution ($\Tilde{S}$, $\Tilde{I}$, $\Tilde{R}$), one can set the determinant of the coefficient matrix as zero, \emph{i.e.},
\begin{equation} \nonumber
    \left |
    \begin{matrix}
    \sigma + k^2 (D_S + \Gamma_S S_{\rm hom}^* w_{SS} Y_k) & \lambda S_{\rm hom}^* + k^2 \Gamma_S S_{\rm hom}^* w_{SI} Y_k & k^2\Gamma_S S_{\rm hom}^* w_{SR} Y_k \\
    0 & \sigma + k^2 D_I - (\lambda S_{\rm hom}^* - \gamma) & 0 \\
    k^2\Gamma_R R_{\rm hom}^* w_{RS} Y_k & -\gamma + k^2 \Gamma_R R_{\rm hom}^* w_{RI} Y_k & \sigma + k^2 (D_R + \Gamma_R R_{\rm hom}^* w_{RR} Y_k) \\
    \end{matrix}
    \right |
    = 0.
\end{equation}
This cubic equation for $\sigma$ could be rewritten as
\[
l(\sigma) q(\sigma) = 0, 
\]
where
\[
        l(\sigma) = \sigma + k^2 D_I  - \lambda S_{\rm hom}^* + \gamma ~\text{ and }~
        q(\sigma) = \sigma^2 + Q_B(k)\sigma +Q_C(k),
\]
with 
\[
    \begin{split}
        &Q_B(k) = k^2 \left[ D_S+D_R+ \left(\Gamma_S S_{\rm hom}^* w_{SS} +  \Gamma_R R_{\rm hom}^* w_{RR}   \right)Y_k \right], \\
        &Q_C(k) = k^4 [ D_S D_R + \left( D_S  \Gamma_R  R_{\rm hom}^* w_{RR} + D_R  \Gamma_S S_{\rm hom}^* w_{SS}  \right)Y_k \\
         & \quad \qquad \qquad+ \Gamma_S \Gamma_R S_{\rm hom}^* R_{\rm hom}^* \left(  w_{SS}w_{RR}-w_{SR}w_{RS}  \right)Y_k^2 ].
    \end{split}
\]
Denote the roots of $l(\sigma)$ and $q(\sigma)$ as $\sigma_1(k) = -D_I k^2 + \lambda S_{\rm hom}^* - \gamma$ and  $\sigma_{2,\pm}(k)$, respectively. Therefore, the steady state solution \reff{SIR DH model s.s. eq of N approx} to the DFT-SIR-DH model is linearly unstable if $\sigma_1(k)$ or $\sigma_{2,\pm}(k)$ have positive real part for some $k$. If $\sigma_1(0) > 0$, then ${\rm Re}(\sigma_1(k)) > 0$ for some $k$, by the continuity of $\sigma_1(\cdot)$. This leads to the first condition \reff{SIR DH lin.stb. condition Reaction}. On the other hand, 
${\rm Re}(\sigma_{2,\pm}(k)) > 0$ for some $k$ if $Q_B(k)<0$ or $Q_C(k)<0$ for some $k$. Let 
\[
Q_b(k):= D_S+D_R+ \left(\Gamma_S S_{\rm hom}^* w_{SS} +  \Gamma_R R_{\rm hom}^* w_{RR}   \right)Y_k.
\]
By the continuity of $Q_b(k)$ at $k=0$, one can obtain that if $Q_b(0)<0$, \emph{i.e.,} (\ref{SIR DH lin.stb. condition Energy QB}), then $Q_B(k)=k^2 Q_b(k)<0$ for some $k$. 
% then $Q_B(k)<0$ for some $k$ if $Q_b(k)<0$ for some $k$.
% Since $Q_b(\infty) = D_S+D_R > 0$, then $Q_b(k)<0$ for some $k$ if $Q_b(0)<0$,  holds.
% If $Q_C(k)<0$ for some $k$, 
Let
\begin{equation}\nonumber
\begin{split}
    Q_c(k) &:= D_S D_R + \left( D_S  \Gamma_R  R_{\rm hom}^* w_{RR} + D_R  \Gamma_S S_{\rm hom}^* w_{SS}  \right)Y_k \\
         &\qquad + \Gamma_S \Gamma_R S_{\rm hom}^* R_{\rm hom}^* \left(  w_{SS}w_{RR}-w_{SR}w_{RS}  \right)Y_k^2.
\end{split}
\end{equation}
Similarly, by the continuity of $Q_c(k)$ at $k=0$, one can obtain that if $Q_c(0)<0$, \emph{i.e.,} (\ref{SIR DH lin.stb. condition Energy QC}), then $Q_C(k)=k^4 Q_c(k)<0$ for some $k$. This completes the proof.
% Similar to the former case, $Q_C(k)<0$ for some $k$ if $Q_c(0)<0$, \emph{i.e.} (\ref{SIR DH lin.stb. condition Energy QC}) holds.
\end{proof}

\begin{remark}
The first condition \reff{SIR DH lin.stb. condition Reaction} indicates that a disease outbreaks when $\lambda S_{\rm hom}^*/\gamma > 1$, which is related to the basic reproduction number defined by, e.g., $\mathcal{R}_0 := \lambda S_{\rm hom}^*/\gamma$~\cite{ReproductionNumber1, ReproductionNumber2}. The second condition \reff{SIR DH lin.stb. condition Energy QB} only possibly holds when some of the coefficients $c_l$ are negative, corresponding to a decomposition with some attractive Yukawa potentials in \reff{SOE-U}. For repulsive potentials ($c_l>0$ for $l=1, \cdots, N$) that are used to model the practice of social distancing in epidemic dynamics, the second condition \reff{SIR DH lin.stb. condition Energy QB} would not hold. It is of interest to note from the third condition \reff{SIR DH lin.stb. condition Energy QC}, which implies $w_{SR}w_{RS} > w_{SS} w_{RR}$ for positive $c_l$, that the epidemic would grow when interaction intensity between $S$ and $R$ is relatively stronger than intra-action intensity. 
% is helpful to explain the disease outbreak. For example, suppose the individuals in some region $\Omega$ are susceptible with total number $N$. Assume the susceptible individuals is distributed homogeneously, i.e. $S(\bm{x},t) \equiv S^0_{\rm hom} >0$. Define the basic reproduction number of (\ref{SIR DH model}) as $\mathcal{R}_0 := \frac{\lambda S^0_{\rm hom}}{\gamma}$. Then by Theorem \ref{Thm. SIR DH lin.stb. condition}, the disease outbreak criterion is $\mathcal{R}_0 > 1$. Since
% \[
% \mathcal{R}_0 = \frac{\lambda N}{\gamma |\Omega|} = \frac{\Bar{\lambda} N}{\gamma} \enspace \text{with} \enspace\Bar{\lambda} = \frac{\lambda}{|\Omega|},
% \]
% the disease outbreak criterion of SIR Debye-Huckel model coincides with classical SIR model \cite{murray_Springer_2002_vol1}. 
\end{remark}

\begin{remark}
The instability conditions derived in the Theorem \ref{t:LBA} can be categorized according to the work~\cite{aslyamov_PRL23}. The  first condition \reff{SIR DH lin.stb. condition Reaction} only involves reaction terms, while the second condition \reff{SIR DH lin.stb. condition Energy QB} and third condition \reff{SIR DH lin.stb. condition Energy QC} are related to the DFT free energy \reff{F}. Therefore, as categorized in~\cite{aslyamov_PRL23},  \reff{SIR DH lin.stb. condition Reaction} can be classified as the R-type, and \reff{SIR DH lin.stb. condition Energy QB} and \reff{SIR DH lin.stb. condition Energy QC} can be regarded as the E-type. 
\end{remark}

\section{Global existence} \label{Sec. PDE analysis}
To further understand the newly derived DFT-SIR-DH model~\reff{DH-SIR-DFT}, we study the global existence of the solution to the model with a simplified self-repulsive potential.   More specifically, we assume that $N=1$, $c_1=1$, $V_S^{\rm ext}=V_I^{\rm ext}=V_R^{\rm ext}=0$, and $W$ is a diagonal matrix with positive elements.  Then, we shall analyze the following initial-boundary value problem
\begin{equation} \label{IBVP}
    \left\{
        \begin{split}
            &\partial_t S = D_S \Delta S + \Gamma_S w_{SS} \nabla \cdot (S\nabla \psi_S ) -\lambda SI, \\
            &\partial_t I = D_I \Delta I + \Gamma_I w_{II} \nabla \cdot (I\nabla \psi_I  ) +\lambda SI -\gamma I,  \\
            &\partial_t R = D_R \Delta R + \Gamma_R  w_{RR} \nabla \cdot (R\nabla \psi_R) +\gamma I, \\
            & -\Delta \psi_S + \kappa^2 \psi_S =  S,  \\
            & -\Delta \psi_I + \kappa^2 \psi_I =  I,  \\
            & -\Delta \psi_R + \kappa^2 \psi_R =  R,  \\
            & S(\bm{x},0) = S_0(\bm{x}), I(\bm{x},0) = I_0(\bm{x}), R(\bm{x},0) = R_0(\bm{x}), \\
            &\nabla S \cdot \Vec{\bm{n}} = \nabla I \cdot \Vec{\bm{n}} = \nabla R \cdot \Vec{\bm{n}} = \nabla \psi_S \cdot \Vec{\bm{n}} = \nabla \psi_I \cdot \Vec{\bm{n}} = \nabla \psi_R \cdot \Vec{\bm{n}} = 0 \enspace \text{on} \enspace \partial \Omega, \\
        \end{split}
    \right.
\end{equation}
where $\Omega \subset \R^n$ is an open bounded domain with smooth boundary $\partial \Omega$.
% In the rest of this work, we assume that $c_l$ for $l=1, \cdots, N$ are nonnegative and take positive elements in the matrix $W$, to model the practice of social distancing in epidemic dynamics.

We first introduce the following basic result on local existence with a relevant extensibility
criterion, which can be similarly proved by an application of the standard parabolic theory to the chemotaxis-type systems. The proof is omitted here; See \emph{e.g.},~\cite{amann_MathZ89, winkler_JDE05, TaoWinkler_NARWA23}. 

\begin{lemma}[Local existence]\label{lemmaLocal}
Assume that $S_0, I_0, R_0 \in W^{1,\infty}(\Omega)$ are non-negative. Then there exist $T_{\max} \in (0,\infty]$ and unique non-negative functions $S,I,R,\psi_S,\psi_I,\psi_R \in C^0(\Bar{\Omega} \times [0,T_{\max})) \cap C^{2,1}(\Bar{\Omega} \times (0,T_{\max}))$ solving \reff{IBVP}  classically in $\Omega \times (0,T_{\max})$. In addition, if $T_{\max} < \infty$, then
\[
\|S(\cdot,t)\|_{L^{\infty}(\Omega)}, \|I(\cdot,t)\|_{L^{\infty}(\Omega)}, \|R(\cdot,t)\|_{L^{\infty}(\Omega)} \rightarrow \infty,
\]
as $t \rightarrow T_{\max}$. 
\end{lemma}
For simplicity, the integral $\int_{\Omega} \rho(\bm{x},t) {\rm d}\x$ is written as $\int_{\Omega} \rho$. The following important properties can be readily derived. 
\begin{lemma}[Mass conservation and positivity]\label{Mass+}
The solution $(S,I,R)$ of \reff{IBVP} satisfies the following property
\begin{equation} \label{SIR DH model mass conserv. eq}
    \int_{\Omega} S+I+R = \int_{\Omega} S_0+I_0+R_0 \enspace \text{for all} \enspace t \in (0,T_{\max}).
\end{equation}
Moreover, if $S_0,I_0,R_0 \ge 0 \enspace \text{for all} \enspace \bm{x} \in \Omega$, then the corresponding solution $(S,I,R,\psi_S,\psi_I,\psi_R )$ of \reff{IBVP} satisfies
\begin{equation} \label{SIR DH model posi. pre. for SIR}
\begin{aligned}
        &S,I,R,\psi_S, \psi_I, \psi_R \ge 0 \enspace \text{for all} \enspace \bm{x} \in \Omega, t \in (0,T_{\max}), \\
        &\int_{\Omega} S, \int_{\Omega} I, \int_{\Omega} R \le M \enspace \text{for all} \enspace t \in (0,T_{\max}), \\
        &\int_{\Omega} \psi_S, \int_{\Omega} \psi_I, \int_{\Omega} \psi_R \le 
\frac{M}{\kappa^2} \enspace \text{for all} \enspace t \in (0,T_{\max}),
\end{aligned}
% Additionally, we  have the inequality for the corresponding solution $\psi_S,\psi_I,\psi_R$ of (\ref{SIR DH model})
% and corresponding solution $\psi_S,\psi_I,\psi_R$  satisfies
% \begin{equation} \label{SIR DH model posi. pre. for psi}
%     \psi_S, \psi_I, \psi_R \ge 0 \enspace \text{for all} \enspace \bm{x} \in \Omega, t \in (0,T_{\max}).
\end{equation}
where $M:=\int_{\Omega} S_0+I_0+R_0 >0$ is the total mass.
\end{lemma}
% With Theorem \ref{Thm. SIR DH mass and positive property}, we can easily derive the following lemma.
% \begin{lemma} \label{Lem. L1 estimate}
% Suppose $S_0,I_0,R_0$ are non-negative functions, then the corresponding solution $S,I,R$ of \reff{IBVP} satisfies the following property
% \[
% \int_{\Omega} S, \int_{\Omega} I, \int_{\Omega} R \le M \enspace \text{for all} \enspace t \in (0,T_{\max}),
% \]
% where $M:=\int_{\Omega} S_0+I_0+R_0 >0$ is a constant.
% Additionally, we also have the inequality for the corresponding solution $\psi_S,\psi_I,\psi_R$ of (\ref{SIR DH model})
% \[
% \int_{\Omega} \psi_S, \int_{\Omega} \psi_I, \int_{\Omega} \psi_R \le 
% \frac{M}{\kappa^2} \enspace \text{for all} \enspace t \in (0,T_{\max}).
% \]
% \end{lemma}
With such $L^1$ estimates, the proof in the work~\cite[Lemma 4.1]{tao_Math_M3AS13} can be used with slight modification to 
establish the following estimate on $\psi_S, \psi_I$, and $\psi_R$. The proof is thus omitted for brevity. 
\begin{lemma} \label{Lem. psi controlled by rho}
Let $\Omega$ be an open bounded domain in $\R^n$  with smooth boundary $\partial \Omega$. Suppose $S_0, I_0, R_0$ are non-negative functions. Then for any $p > \max(\frac{n}{2},1)$ and $\varepsilon > 0$, there exists some constant $c(\varepsilon)>0$ such that the solution $(S,I,R,\psi_S,\psi_I,\psi_R)$ of \reff{IBVP} satisfies
\[
\int_{\Omega} \psi_{\rho}^{p+1} \le \varepsilon \int_{\Omega} \rho^{p+1} + c(\varepsilon) \enspace \text{for} ~\rho = S,I,R.
\]
% with any $\varepsilon > 0$,  and $\rho = S,I,R.$
\end{lemma}

We next present the following $L^p$ estimates on $S, I$, and $R$, which play a crucial role in the proof of global existence.
\begin{theorem}[$L^p$ estimate]\label{Thm. SIR DH Lp estimate}
Let $\Omega$ be an open bounded domain in $\R^n$  with smooth boundary $\partial \Omega$. Suppose $S_0,I_0,R_0$ be non-negative functions. Then for any $p > \max(\frac{n}{2},1),$
there exists a constant $C>0$ such that the corresponding solution $(S,I,R)$ of (\ref{IBVP}) satisfies
\[
\int_{\Omega} S^p , \int_{\Omega} I^p, \int_{\Omega} R^p \le C \enspace \text{for all} \enspace t \in (0,T_{\max}).
\]
\end{theorem}
\begin{proof}
Testing both sides of the first equation in (\ref{IBVP}) with $S^{p-1}$
% as a test function for the first equation in (\ref{IBVP}) 
and employing the fourth equation in (\ref{IBVP}), we have by integration by parts that 
\begin{equation} \nonumber
    \begin{split}
        \frac{1}{p} \frac{d}{dt} \int_{\Omega} S^p &= D_S \int_{\Omega} S^{p-1} \Delta S + \Gamma_S w_{SS} \int_{\Omega} S^{p-1} \nabla \cdot (S\nabla \psi_S) -\lambda \int_{\Omega} S^p I \\
        &= -D_S (p-1) \int_{\Omega} S^{p-2} |\nabla S|^2 - \Gamma_S w_{SS} (p-1) \int_{\Omega} S^{p-1} \nabla S \cdot \nabla \psi_S -\lambda \int_{\Omega} S^p I \\
        &= -\frac{4 D_S (p-1)}{p^2} \int_{\Omega} |\nabla S^{\frac{p}{2}}|^2 + \frac{\Gamma_S w_{SS} (p-1)}{p} \int_{\Omega} S^p (\kappa^2 \psi_S -S) -\lambda \int_{\Omega} S^p I
    \end{split}
\end{equation}
for all $t \in (0,T_{\max})$. By Lemma~\ref{Mass+}, we have 
\begin{equation} \label{SIR DH model S estimate ineq 1}
    \frac{d}{dt} \int_{\Omega} S^p \le \Gamma_S w_{SS} (p-1) \int_{\Omega} S^p (\kappa^2 \psi_S -S) \enspace \text{for all} \enspace t \in (0,T_{\max}).
\end{equation}
By the Young inequality, we have
\[
\int_{\Omega} S^p \psi_S \le \int_{\Omega} \varepsilon_1 S^{p+1} + (\varepsilon_1 \frac{p+1}{p})^{-p} (p+1)^{-1} \int_{\Omega} \psi_S^{p+1} \enspace \text{for all} \enspace t \in (0,T_{\max}). 
\]
Then by (\ref{SIR DH model S estimate ineq 1}), we arrive at
\begin{equation} \label{SIR DH model S estimate ineq 2}
    \frac{d}{dt} \int_{\Omega} S^p \le -\frac{\Gamma_S w_{SS} (p-1)}{2} \int_{\Omega} S^{p+1} + C_1 \int_{\Omega} \psi_S^{p+1} \enspace \text{for all} \enspace t \in (0,T_{\max}),
\end{equation}
where we have taken $\varepsilon_1 = \frac{1}{2\kappa^2}$ and $C_1 = \Gamma_S w_{SS} (p-1) \kappa^2 (\varepsilon_1 \frac{p+1}{p})^{-p} (p+1)^{-1}. $
By Lemma \ref{Lem. psi controlled by rho}, we have  
\[
\int_{\Omega} \psi_S^{p+1} \le \varepsilon_2 \int_{\Omega} S^{p+1} + C_2(\varepsilon_2) \enspace \text{for all} \enspace t \in (0,T_{\max}).
\]
Then, it follows from (\ref{SIR DH model S estimate ineq 2}) that
\begin{equation} \label{SIR DH model S estimate ineq 3}
    \frac{d}{dt} \int_{\Omega} S^p \le -\frac{\Gamma_S w_{SS} (p-1)}{4} \int_{\Omega} S^{p+1} + C_3 \enspace \text{for all} \enspace t \in (0,T_{\max}), 
\end{equation}
where we have taken $\varepsilon_2 = \frac{\Gamma_S w_{SS} (p-1)}{4 C_1}$ and $C_3 = C_1 C_2(\varepsilon_2). $
Adding the term $\int_{\Omega} S^p$ to both sides of (\ref{SIR DH model S estimate ineq 3}) yields
\[
\frac{d}{dt} \int_{\Omega} S^p + \int_{\Omega} S^p \le -\frac{\Gamma_S w_{SS} (p-1)}{4} \int_{\Omega} S^{p+1} + C_3 + \int_{\Omega} S^p \enspace \text{for all} \enspace t \in (0,T_{\max}).
\]
By the Young inequality, we have
\[
\int_{\Omega} S^p \le \int_{\Omega} \varepsilon_3 S^{p+1} + (\varepsilon_3 \frac{p+1}{p})^{-p} (p+1)^{-1} |\Omega| \enspace \text{for all} \enspace t \in (0,T_{\max}). 
\]
Therefore, we have
\begin{equation} \label{SIR DH model S estimate ineq 4}
    \frac{d}{dt} \int_{\Omega} S^p + \int_{\Omega} S^p \le C_4 \enspace \text{for all} \enspace t \in (0,T_{\max}), 
\end{equation}
where we have taken $\varepsilon_3 = \frac{\Gamma_S w_{SS} (p-1)}{4}$ and $C_4 = C_3 + (\varepsilon_3 \frac{p+1}{p})^{-p} (p+1)^{-1} |\Omega|. $
The inequality (\ref{SIR DH model S estimate ineq 4}), together with the Gronwall inequality, yields
\begin{equation} \label{SIR DH model S Lp bdd.}
    \int_{\Omega} S^p \le e^{-t} \int_{\Omega} S_0^p + C_4(1-e^{-t}) \le \int_{\Omega} S_0^p + C_4 := C_S(p) \enspace \text{for all} \enspace t \in (0,T_{\max}). 
\end{equation}
Similarly, using $I^{p-1}$ as a test function for the second equation in (\ref{IBVP}), integrating by parts and employing the fifth equation in (\ref{IBVP}), we have
\begin{equation} \nonumber
    \begin{split}
        \frac{1}{p} \frac{d}{dt} \int_{\Omega} I^p &= D_I \int_{\Omega} I^{p-1} \Delta I + \Gamma_I w_{II} \int_{\Omega} I^{p-1} \nabla \cdot (I\nabla \psi_I) + \lambda \int_{\Omega} S I^p - \gamma \int_{\Omega} I^p \\
        &= -D_I (p-1) \int_{\Omega} I^{p-2} |\nabla I|^2 - \Gamma_I w_{II} (p-1) \int_{\Omega} I^{p-1} \nabla I \cdot \nabla \psi_I + \lambda \int_{\Omega} S I^p - \gamma \int_{\Omega} I^p \\
        &= -\frac{4 D_I (p-1)}{p^2} \int_{\Omega} |\nabla I^{\frac{p}{2}}|^2 + \frac{\Gamma_I w_{II} (p-1)}{p} \int_{\Omega} I^p (\kappa^2 \psi_I -I) + \lambda \int_{\Omega} S I^p - \gamma \int_{\Omega} I^p
    \end{split}
\end{equation}
for all $t \in (0,T_{\max})$. By Lemma \ref{Mass+}, we have 
\begin{equation} \label{SIR DH model I estimate ineq 1}
    \frac{d}{dt} \int_{\Omega} I^p \le \Gamma_I w_{II} (p-1) \int_{\Omega} I^p (\kappa^2 \psi_I -I) + p \lambda \int_{\Omega} S I^p \enspace \text{for all} \enspace t \in (0,T_{\max}).
\end{equation}
By the Young inequality, we have
\begin{equation} \label{SIR DH model I estimate ineq 2}
    \begin{split}
        & \int_{\Omega} S I^p \le \int_{\Omega} \varepsilon_4 I^{p+1} + (\varepsilon_4 \frac{p+1}{p})^{-p} (p+1)^{-1} \int_{\Omega} S^{p+1}, \\
        & \int_{\Omega} I^p \psi_I \le \int_{\Omega} \varepsilon_5 I^{p+1} + (\varepsilon_5 \frac{p+1}{p})^{-p} (p+1)^{-1} \int_{\Omega} \psi_I^{p+1}
    \end{split}
\end{equation}
for all $t \in (0,T_{\max})$. Combining (\ref{SIR DH model I estimate ineq 1}) and (\ref{SIR DH model I estimate ineq 2}), we arrive at
\begin{equation} \label{SIR DH model I estimate ineq 3}
    \frac{d}{dt} \int_{\Omega} I^p \le -\frac{\Gamma_I w_{II} (p-1)}{2} \int_{\Omega} I^{p+1} + C_5 \int_{\Omega} \psi_I^{p+1} + C_6 \int_{\Omega} S^{p+1} \enspace \text{for all} \enspace t \in (0,T_{\max}),
\end{equation}
where $\varepsilon_4 = \frac{\Gamma_I w_{II} (p-1)}{4p \lambda}$, $\varepsilon_5 = \frac{1}{4 \kappa^2}$, $C_5 = \Gamma_I w_{II} (p-1) \kappa^2 (\varepsilon_5 \frac{p+1}{p})^{-p} (p+1)^{-1}$, and $C_6 = p \lambda (\varepsilon_4 \frac{p+1}{p})^{-p} (p+1)^{-1}. $
By Lemma \ref{Lem. psi controlled by rho}, we have 
\[
\int_{\Omega} \psi_I^{p+1} \le \varepsilon_6 \int_{\Omega} I^{p+1} + C_7(\varepsilon_6) \enspace \text{for all} \enspace t \in (0,T_{\max}).
\]
Combining the above estimates, (\ref{SIR DH model S Lp bdd.}) and (\ref{SIR DH model I estimate ineq 3}), we have
\begin{equation} \label{SIR DH model I estimate ineq 4}
    \frac{d}{dt} \int_{\Omega} I^p \le -\frac{\Gamma_I w_{II} (p-1)}{4} \int_{\Omega} I^{p+1} + C_8 \enspace \text{for all} \enspace t \in (0,T_{\max}), 
\end{equation}
where we have taken $\varepsilon_6 = \frac{\Gamma_I w_{II} (p-1)}{4 c_5}$ and $C_8 = C_5 C_7(\varepsilon_6) + C_6 C_S(p+1). $
Adding the term $\int_{\Omega} I^p$ to both sides of (\ref{SIR DH model I estimate ineq 4}) yields
\[
\frac{d}{dt} \int_{\Omega} I^p + \int_{\Omega} I^p \le -\frac{\Gamma_I w_{II} (p-1)}{4} \int_{\Omega} I^{p+1} + C_8 + \int_{\Omega} I^p \enspace \text{for all} \enspace t \in (0,T_{\max}).
\]
By the Young inequality, we have
\[
\int_{\Omega} I^p \le \int_{\Omega} \varepsilon_7 I^{p+1} + (\varepsilon_7 \frac{p+1}{p})^{-p} (p+1)^{-1} |\Omega| \enspace \text{for all} \enspace t \in (0,T_{\max}). 
\]
Therefore, we have
\begin{equation} \label{SIR DH model I estimate ineq 5}
    \frac{d}{dt} \int_{\Omega} I^p + \int_{\Omega} I^p \le C_9 \enspace \text{for all} \enspace t \in (0,T_{\max}), 
\end{equation}
where we have taken $\varepsilon_7 = \frac{\Gamma_I w_{II} (p-1)}{4}$ and $C_9 = C_8 + (\varepsilon_7 \frac{p+1}{p})^{-p} (p+1)^{-1} |\Omega|. $
By the Gronwall inequality, (\ref{SIR DH model I estimate ineq 5}) yields
\begin{equation} \label{SIR DH model I Lp bdd.}
    \int_{\Omega} I^p \le e^{-t} \int_{\Omega} I_0^p + C_9(1-e^{-t}) \le \int_{\Omega} I_0^p + C_9 := C_I(p) \enspace \text{for all} \enspace t \in (0,T_{\max}).
\end{equation}
Analogously, using $R^{p-1}$ as a test function for the third equation in (\ref{IBVP}), integrating by parts and employing the sixth equation in (\ref{IBVP}), we have
\begin{equation} \nonumber
    \begin{split}
        \frac{1}{p} \frac{d}{dt} \int_{\Omega} R^p &= D_R \int_{\Omega} R^{p-1} \Delta R + \Gamma_R w_{RR} \int_{\Omega} R^{p-1} \nabla \cdot (R\nabla \psi_R) +\gamma \int_{\Omega} I R^{p-1} \\
        &= -D_R (p-1) \int_{\Omega} R^{p-2} |\nabla R|^2 - \Gamma_R w_{RR} (p-1) \int_{\Omega} R^{p-1} \nabla R \cdot \nabla \psi_R +\gamma \int_{\Omega} I R^{p-1} \\
        &= -\frac{4 D_R (p-1)}{p^2} \int_{\Omega} |\nabla R^{\frac{p}{2}}|^2 + \frac{\Gamma_R w_{RR} (p-1)}{p} \int_{\Omega} R^p (\kappa^2 \psi_R -R) +\gamma \int_{\Omega} I R^{p-1}
    \end{split}
\end{equation}
for all $t \in (0,T_{\max})$. Then we arrive at
\begin{equation} \label{SIR DH model R estimate ineq 1}
    \frac{d}{dt} \int_{\Omega} R^p \le \Gamma_R w_{RR} (p-1) \int_{\Omega} R^p (\kappa^2 \psi_R -R) + p \gamma \int_{\Omega} I R^{p-1} \enspace \text{for all} \enspace t \in (0,T_{\max}). 
\end{equation}
By the Young inequality, we have
\begin{equation} \label{SIR DH model R estimate ineq 2}
    \begin{split}
        & \int_{\Omega} I R^{p-1} \le \int_{\Omega} \varepsilon_8 R^p + \int_{\Omega} (\varepsilon_8 \frac{p}{p-1})^{-(p-1)} p^{-1} I^p, \\
        & \int_{\Omega} R^p \psi_R \le \int_{\Omega} \varepsilon_9 R^{p+1} +\int_{\Omega} (\varepsilon_9 \frac{p+1}{p})^{-p} (p+1)^{-1} \psi_R^{p+1}
    \end{split}
\end{equation}
for all $t \in (0,T_{\max})$. Combining (\ref{SIR DH model R estimate ineq 1}) and (\ref{SIR DH model R estimate ineq 2}), we have
\begin{equation} \label{SIR DH model R estimate ineq 3}
    \frac{d}{dt} \int_{\Omega} R^p \le -\frac{\Gamma_R w_{RR} (p-1)}{2} \int_{\Omega} R^{p+1} + C_{10} \int_{\Omega} \psi_R^{p+1} + \int_{\Omega} R^p + C_{11} \int_{\Omega} I^p \enspace \text{for all} \enspace t \in (0,T_{\max}),  
\end{equation}
where we have taken $\varepsilon_8 = \frac{1}{p \gamma}$, $ \varepsilon_9 = \frac{1}{2 \kappa^2}$, $C_{10} = \Gamma_R w_{RR} (p-1) \kappa^2 (\varepsilon_9 \frac{p+1}{p})^{-p} (p+1)^{-1}$, and $C_{11} = p \gamma (\varepsilon_8 \frac{p}{p-1})^{-(p-1)} p^{-1}. $
By Lemma \ref{Lem. psi controlled by rho}, we have  
\[
\int_{\Omega} \psi_R^{p+1} \le \varepsilon_{10} \int_{\Omega} R^{p+1} + C_{12}(\varepsilon_{10}) \enspace \text{for all} \enspace t \in (0,T_{\max}).
\]
Combining the above estimates, (\ref{SIR DH model I Lp bdd.}) and (\ref{SIR DH model R estimate ineq 3}), we arrive at
\begin{equation} \label{SIR DH model R estimate ineq 4}
    \frac{d}{dt} \int_{\Omega} R^p \le -\frac{\Gamma_R w_{RR} (p-1)}{4} \int_{\Omega} R^{p+1} + \int_{\Omega} R^p + C_{13} \enspace \text{for all} \enspace t \in (0,T_{\max}), 
\end{equation}
where $\varepsilon_{10} = \frac{\Gamma_R w_{RR} (p-1)}{4 C_{10}}$ and $C_{13} = C_{10}C_{12}(\varepsilon_{10}) + C_{11} C_I(p). $
Adding the term $\int_{\Omega} R^p$ to both sides of (\ref{SIR DH model R estimate ineq 4}) yields
\[
\frac{d}{dt} \int_{\Omega} R^p  + \int_{\Omega} R^p \le -\frac{\Gamma_R w_{RR} (p-1)}{4} \int_{\Omega} R^{p+1} + 2 \int_{\Omega} R^p + C_{13} \enspace \text{for all} \enspace t \in (0,T_{\max}). 
\]
By the Young inequality, we have
\[
\int_{\Omega} R^p \le \int_{\Omega} \varepsilon_{11} R^{p+1} + (\varepsilon_{11} \frac{p+1}{p})^{-p} (p+1)^{-1} |\Omega| \enspace \text{for all} \enspace t \in (0,T_{\max}). 
\]
Then we have
\begin{equation} \label{SIR DH model R estimate ineq 5}
    \frac{d}{dt} \int_{\Omega} R^p  + \int_{\Omega} R^p \le C_{14} \enspace \text{for all} \enspace t \in (0,T_{\max}), 
\end{equation}
where we have taken $\varepsilon_{11} = \frac{\Gamma_R w_{RR} (p-1)}{8}$ and $C_{14} = C_{13} + 2 (\varepsilon_{11} \frac{p+1}{p})^{-p} (p+1)^{-1} |\Omega|. $
Applying the Gronwall inequality to (\ref{SIR DH model R estimate ineq 5}) yields
\begin{equation} \label{SIR DH model R Lp bdd.}
    \int_{\Omega} R^p \le e^{-t} \int_{\Omega} R_0^p + C_{14}(1-e^{-t}) \le \int_{\Omega} R_0^p + C_{14} := C_R(p) \enspace \text{for all} \enspace t \in (0,T_{\max}).
\end{equation}
This completes the proof.
\end{proof}

Following the technique of the Moser-Alikakos iteration~\cite{Alikakos79, tao_JDE12}, we can further establish the $L^{\infty}$ bounds for $S$, $I$, and $R$. Therefore, the following result on global existence of the solution to the model~\reff{IBVP} can be obtained by applying the Lemma~\ref{lemmaLocal}.
\begin{theorem}[Global existence]
Let $\Omega$ be an open bounded domain in $\R^n$ with a smooth boundary $\partial \Omega$. Suppose $S_0, I_0, R_0 \in W^{1,\infty}(\Omega)$ are non-negative functions. Then 
% for any $p > \max(\frac{n}{2},1),$
there exist unique non-negative bounded functions $S, I, R, \psi_S, \psi_I, \psi_R \in C^0(\Bar{\Omega} \times [0,\infty)) \cap C^{2,1}(\Bar{\Omega} \times (0,\infty))$ solving (\ref{IBVP}) classically.
% Suppose that $W$ is a diagonal matrix and $S_0, I_0, R_0 \in W^{1,\infty}(\Omega)$ are non-negative functions. Then for any $n \ge 2$, there exist unique non-negative bounded functions $S, I, R, \psi_S, \psi_I, \psi_R \in C^0(\Bar{\Omega} \times [0,\infty)) \cap C^{2,1}(\Bar{\Omega} \times (0,\infty))$ solving (\ref{SIR DH model}) classically.
\end{theorem}

\section{Numerical Investigation}\label{Sec. Numerical results}
To solve the nonlinear DFT-SIR-DH model, we propose a structure-preserving numerical scheme that guarantees positivity, mass conservation, and energy dissipation, based on the ones developed in our previous works~\cite{LiuYueWangWiseZhou_MathComput21, DingZhou_JCP24, DingXuZhou_DCSD-B25, DingWangZhou_MathComput25}. The time-splitting schemes~\cite{DingXuZhou_DCSD-B25} that respect energy dissipation are applied to separate the model into a reaction subproblem and a subproblem of Keller-Segal equations (or Poisson--Nernst--Planck equations). Then the subproblem of Keller-Segal equations can be implicitly treated using logarithmic variables to preserve positivity and  energy dissipation simultaneously~\cite{LiuYueWangWiseZhou_MathComput21, DingZhou_JCP24, DingWangZhou_MathComput25}. The numerical schemes will be reported in a separate work in future.

In this section, we first numerically investigate instability phenomenon of the DFT-SIR-DH model~\reff{DH-SIR-DFT} arising from social interactions and infectious reactions.
% We use operator splitting method to find the numerical solution.
Consider a computational domain $\Omega = [0,1] \times [0,1]$. Unless specified otherwise, we take the following parameters and interaction intensity matrix in our numerical studies:
\[
\beta = 1, N=1, c_1=1, \gamma = 1, \kappa_1 = 1, W = 
\begin{pmatrix}
w_{SS}  & 100 & w_{SR} \\
100  & 100 & 100\\
w_{RS} & 100 & w_{RR}
\end{pmatrix},
D_j=1, V_j^{\rm ext}=0, 
\]
for $j=S, I, R$.
% where $\rho = w_{SS} = w_{RR}$ and $\omega = w_{SR} = w_{RS}$. 
The initial conditions are given by
\begin{equation}\label{PertInit}
    S_0(\bm{x}) = S_{\rm hom}^* + \xi(\bm{x}), ~I_0(\bm{x}) = I_{\rm hom}^* + \xi(\bm{x}),~ R_0(\bm{x}) = R_{\rm hom}^* + \xi(\bm{x}), 
\end{equation}
where 
\[
\xi(\bm{x}) = 10^{-3}\max\{ \cos(10\pi x)\cos(10\pi y), 0 \}
\]
is a small perturbation added to the homogeneous steady state
\[
S_{\rm hom}^* = 0.8,~ I_{\rm hom}^* = 0, ~R_{\rm hom}^* = 0.2. 
\]

\subsection{Instability and distribution}
\begin{figure}[htbp]
    \centering
    \subfloat{
    \includegraphics[width=0.33\linewidth]{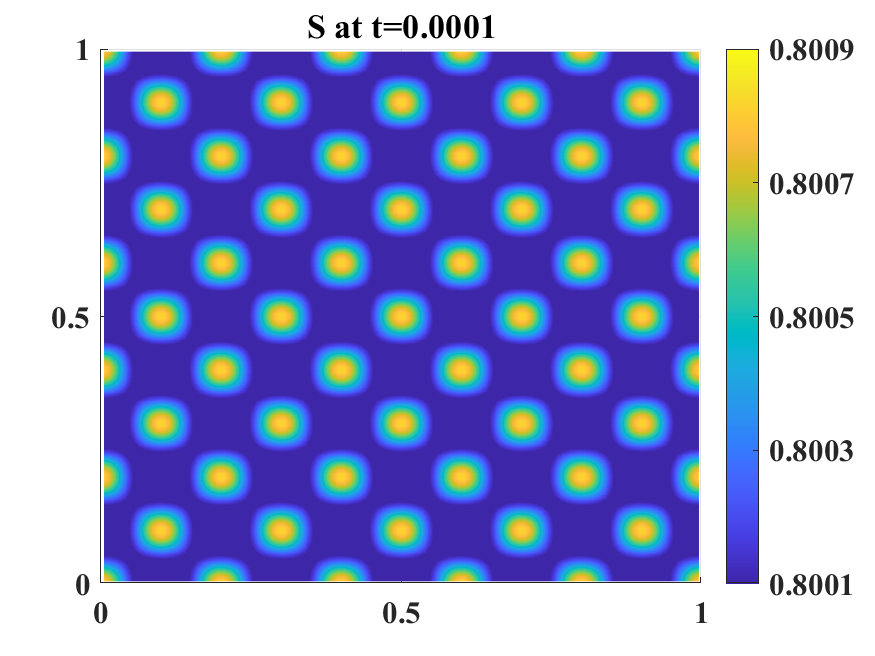}
    }
    \subfloat{
    \includegraphics[width=0.33\linewidth]{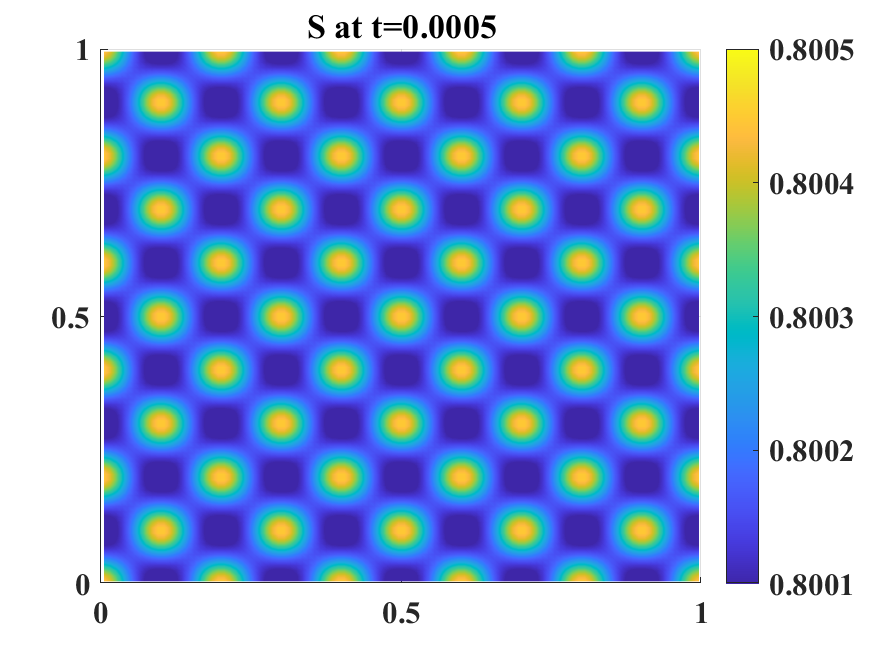}
    }
    \subfloat{
    \includegraphics[width=0.33\linewidth]{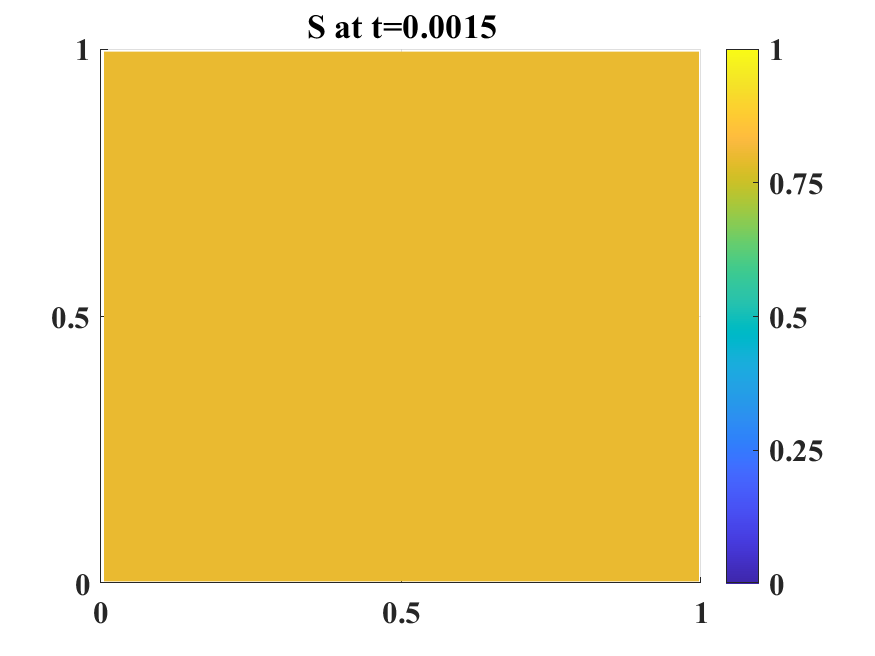}
    }

    \subfloat{
    \includegraphics[width=0.33\linewidth]{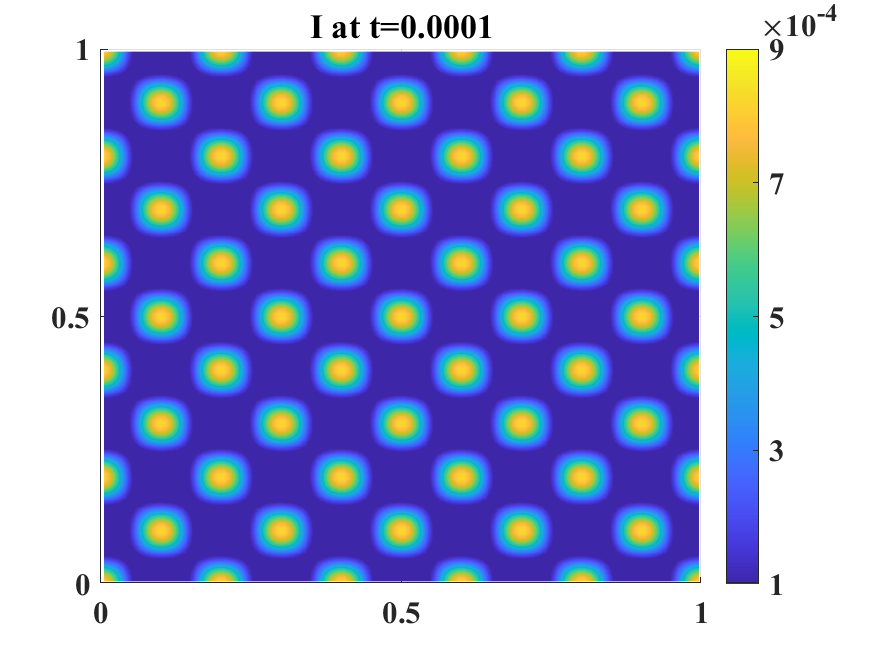}
    }
    \subfloat{
    \includegraphics[width=0.33\linewidth]{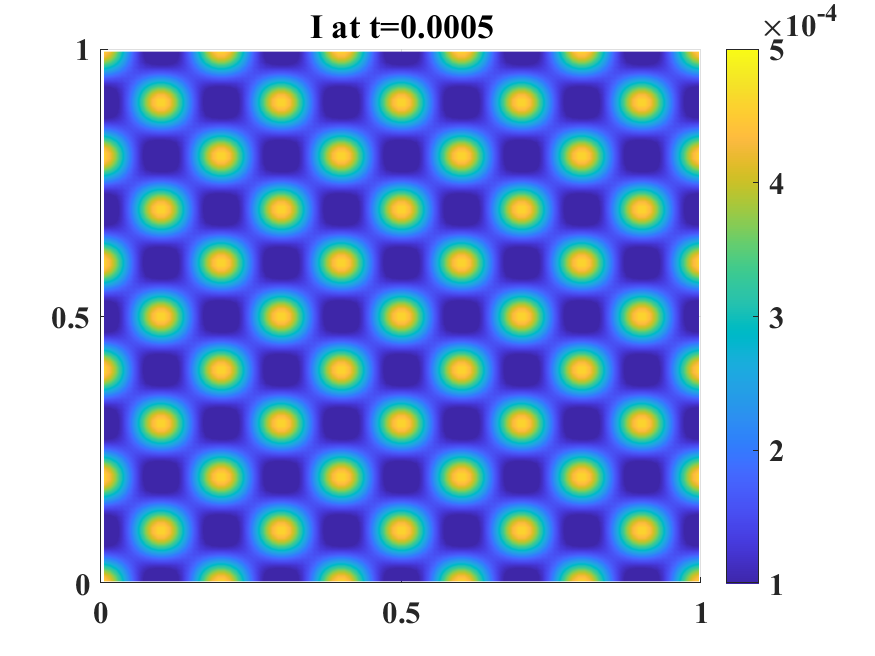}
    }
    \subfloat{
    \includegraphics[width=0.33\linewidth]{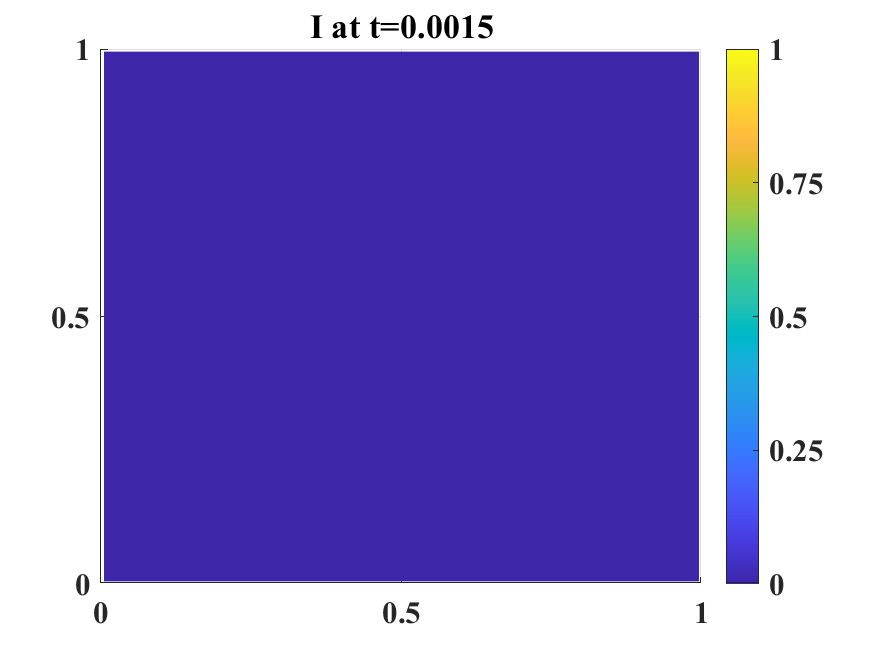}
    }

    \subfloat{
    \includegraphics[width=0.33\linewidth]{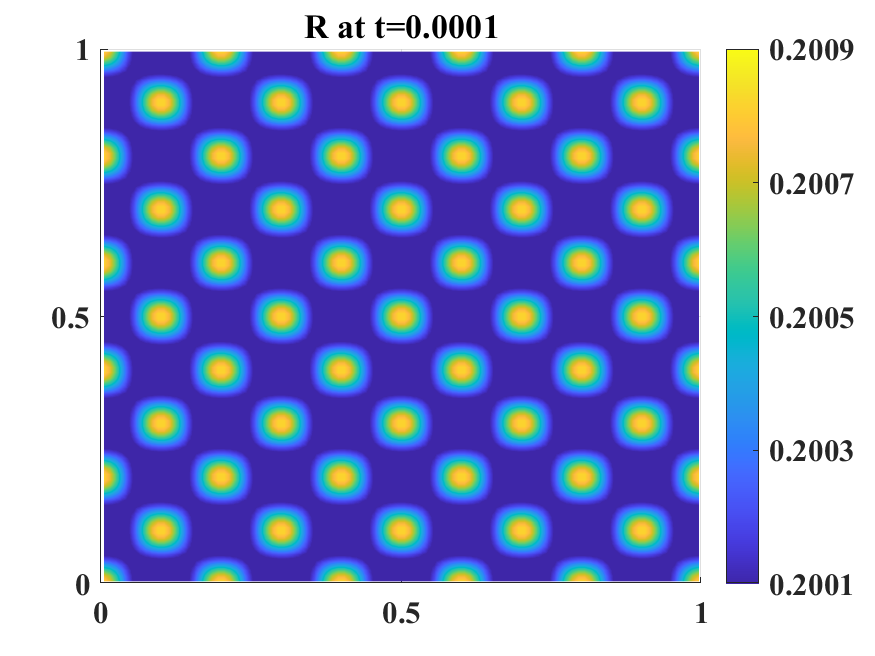}
    }
    \subfloat{
    \includegraphics[width=0.33\linewidth]{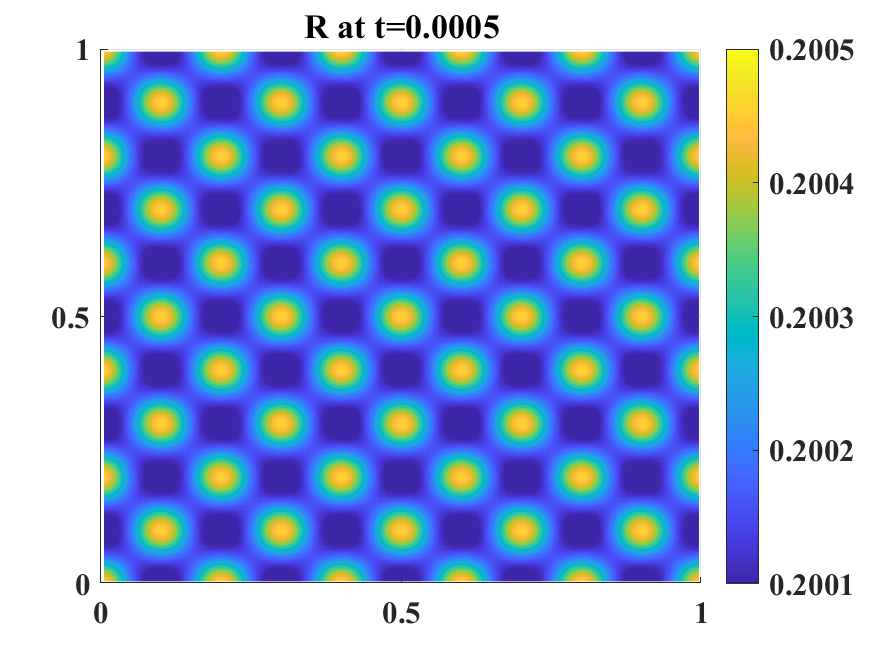}
    }
    \subfloat{
    \includegraphics[width=0.33\linewidth]{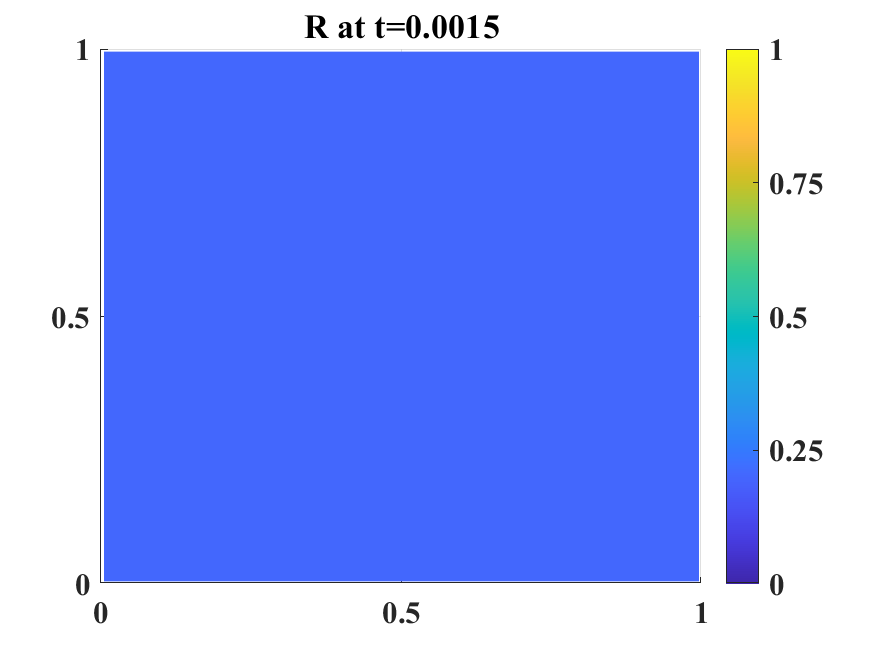}
    }
    \caption{Evolution of $S$, $I$, and $R$ with initial perturbed conditions \reff{PertInit} using $\lambda=1$, $w_{SS} = w_{RR}=1$, and $w_{SR} = w_{RS}=1$.}
    \label{fig:E-type linear stable}
\end{figure}
\begin{figure}[htbp]
    \centering
    \subfloat{
    \includegraphics[width=0.33\linewidth]{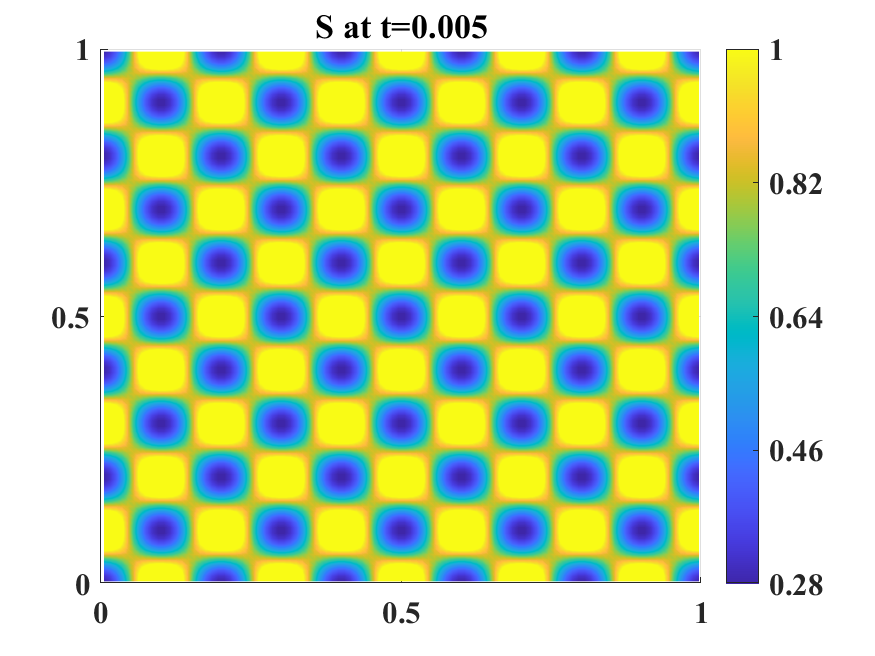}
    }
    \subfloat{
    \includegraphics[width=0.33\linewidth]{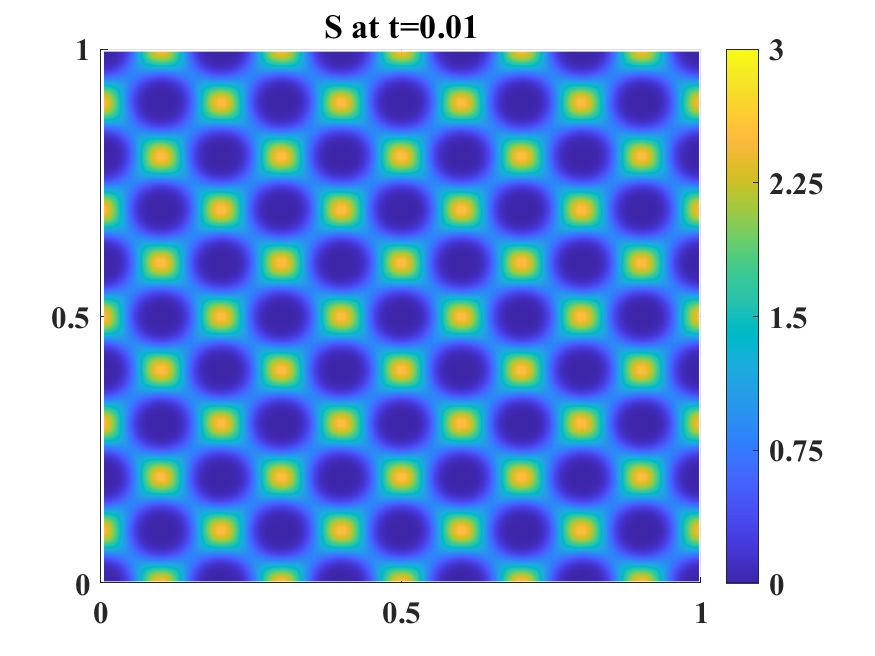}
    }
    \subfloat{
    \includegraphics[width=0.33\linewidth]{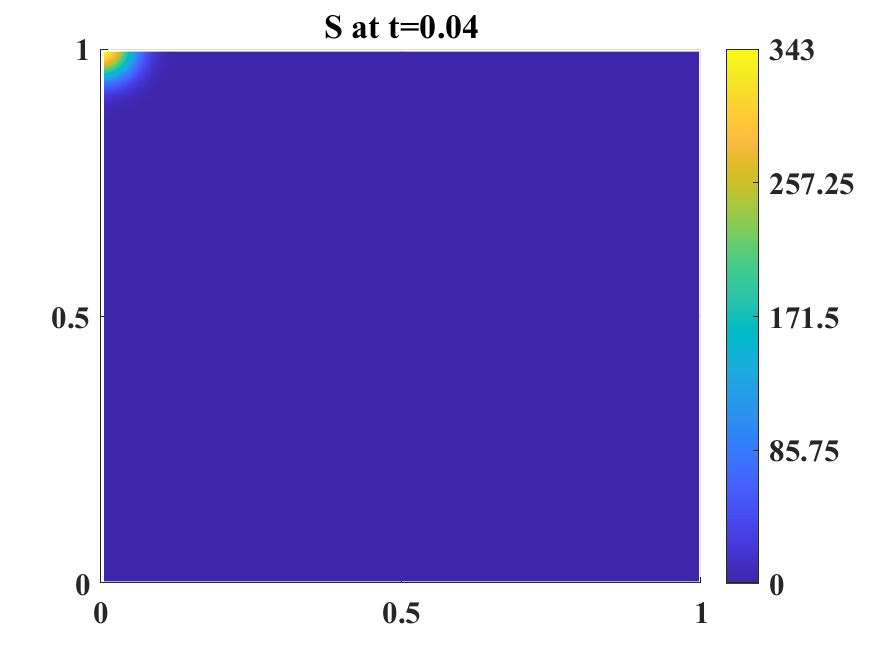}
    }

    \subfloat{
    \includegraphics[width=0.33\linewidth]{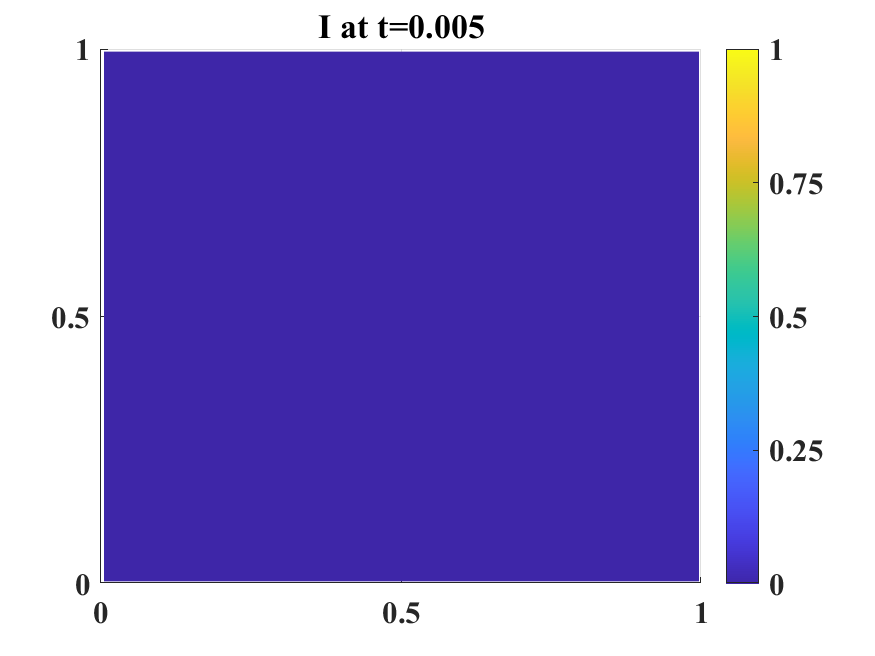}
    }
    \subfloat{
    \includegraphics[width=0.33\linewidth]{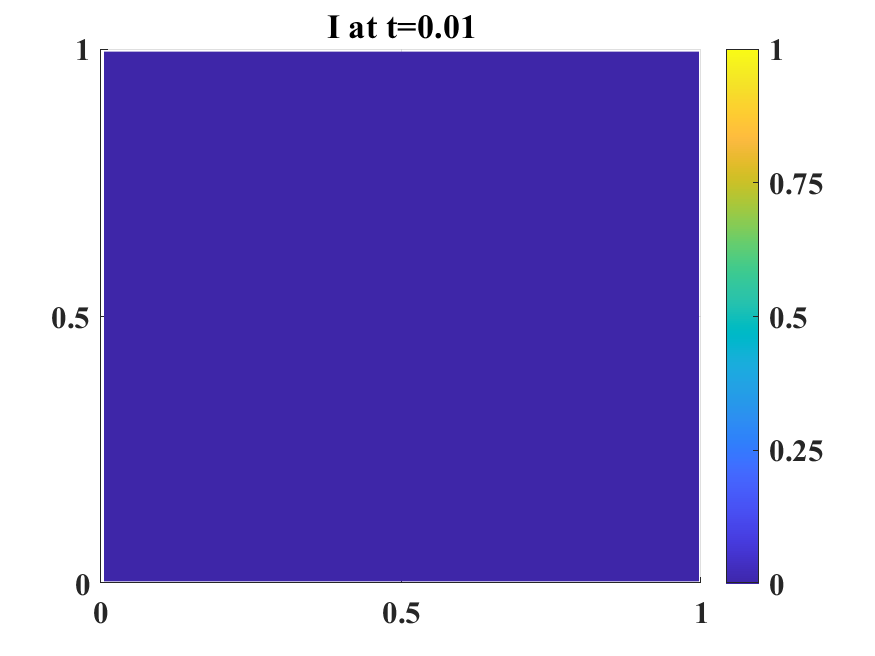}
    }
    \subfloat{
    \includegraphics[width=0.33\linewidth]{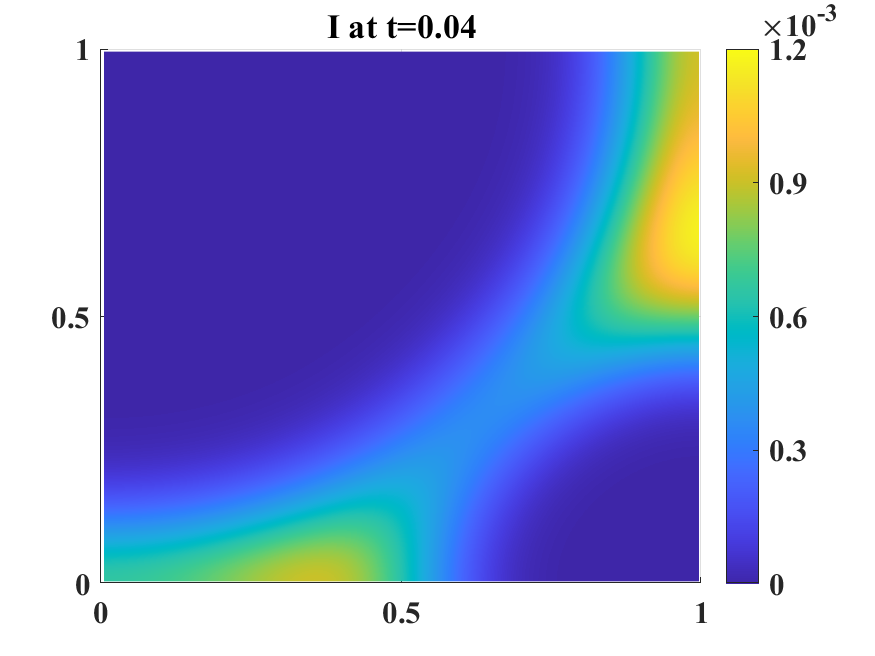}
    }

    \subfloat{
    \includegraphics[width=0.33\linewidth]{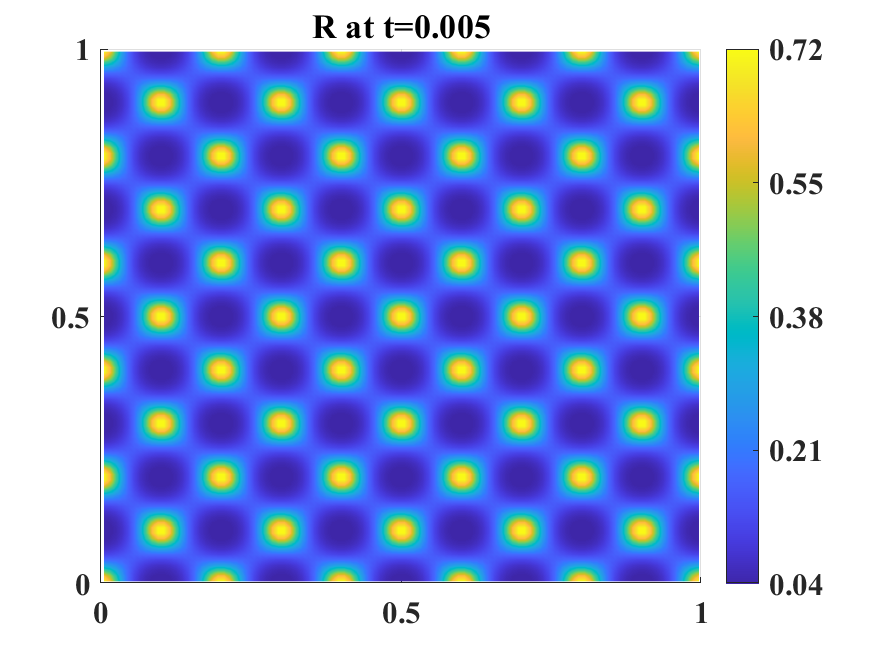}
    }
    \subfloat{
    \includegraphics[width=0.33\linewidth]{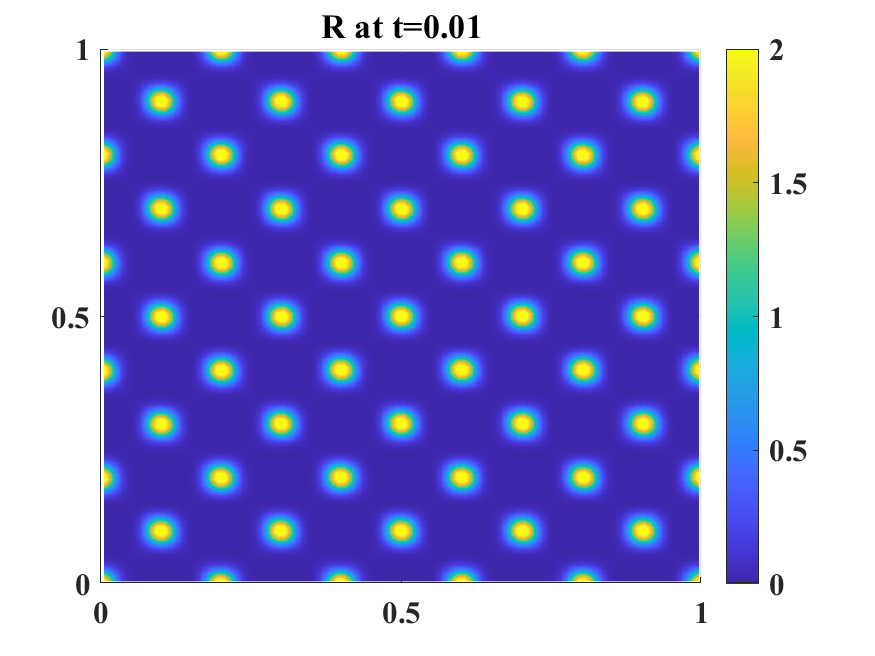}
    }
    \subfloat{
    \includegraphics[width=0.33\linewidth]{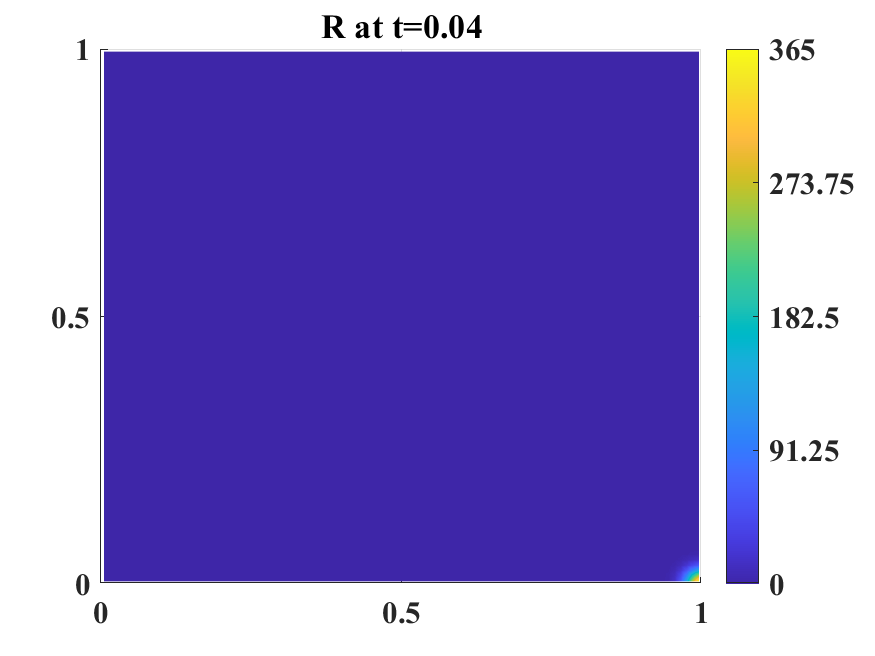}
    }
    \caption{Evolution of $S$, $I$, and $R$ with initial perturbed conditions \reff{PertInit} using $\lambda=1$, $w_{SS} = w_{RR}=1$, and $w_{SR} = w_{RS}=10^4$.}
    \label{fig:E-type linear unstable}
\end{figure}
We now probe different behaviors of instability that takes place according to the conditions revealed by the linear stability analysis in Theorem \ref{t:LBA}. 
% Since the reaction induced instability has been well-studied in literature, we here mainly focus on the instability induced by social interactions. 
First, we fix $\lambda = 1$, $w_{SS} = w_{RR}=1$, and $w_{SR} = w_{RS}=1$, for which neither the reaction induced instability condition \reff{SIR DH lin.stb. condition Reaction}, corresponding to $\mathcal{R}_0 := \lambda S_{\rm hom}^*/\gamma < 1$, nor the social interaction induced instability condition \reff{SIR DH lin.stb. condition Energy QC} holds. Figure~\ref{fig:E-type linear stable} presents the distributions of $S$, $I$, and $R$ at various time $t$. As shown, the perturbed initial densities, as expected from the linear stability analysis in Theorem \ref{t:LBA}, gradually level off and tend to corresponding homogeneous states $S_{\rm hom}$, $I_{\rm hom}^*$, and $R_{\rm hom}^*.$
% In this section, we investigate the E-type instability. Fix $\lambda = 1$ and $\rho = 1$, then ( and $\mathcal{R}_0 := \lambda S_{\rm hom}^*/\gamma < 1$.  Here we present the
% snapshots where $\omega = 1$ and in this case, (\ref{SIR DH lin.stb. condition Energy QC}) is not satisfied.
% As shown in Figure (\ref{fig:E-type linear stable}), the densities tend to the initial homogeneous steady state. This result is consistent with Theorem \reff{Thm. SIR DH of N approx lin.stb. condition} since (\ref{SIR DH lin.stb. condition Energy QC}) is not satisfied and the system is linearly stable. 

We now focus on the instability induced by social interactions. The numerical simulations take the parameters $\lambda = 1$, $w_{SS} = w_{RR}=1$, and $w_{SR} = w_{RS}=10^4$. In such a case, it can be readily verified that the instability condition due to social interactions \reff{SIR DH lin.stb. condition Energy QC} holds. As seen from Figure~\ref{fig:E-type linear unstable}, checker spatial pattern quickly emerges for $S$ due to large cross social interactions~\cite{BenJacob_PhysD85,Banerjee_PhysA20}. As time evolves, the high $S$ density further accumulates and develops isolated dotted peaks to avoid social interactions. In the very end, the susceptible population concentrates in the left upper corner with a large magnitude, reflecting the large repulsive social interactions adopted in simulations. 

For the infected population $I$, the density evolves from the dotted pattern into two peaks locating at right upper and lower left corners due to the social interactions. For the recovered individuals, the density $R$ first develops higher and higher peaks at isolated positions, and eventually all accumulates at the lower right corner with a significantly high concentration. In contrast to the distribution of $S$, it is easy to find that the strong, repulsive social interactions between $S$ and $R$ lead to such extreme distributions of $R$. Overall, one can observe that the individuals of the same kind all finally gather as clusters at corners due to the repulsive interactions under consideration. We need to remark that the location of the corners depends on the initial perturbed conditions~\reff{PertInit}.
% exhibits the densities of the susceptible, infectious and recovered individuals at various time. At the first stage, the density of $S$ forms checkered patterns while the distribution of $R$ forms dotted pattern. As time goes on, the pattern of $S$ shifts to dotted shape. Finally, the individuals of the same kind gather as clusters. Since the clusters repulse each other, then we could observe the phenomenon like phase separation. This is the result of E-type instability.

\subsection{Phase Diagram}
\begin{figure}[htbp]
    \centering
    \includegraphics[width=0.618\linewidth]{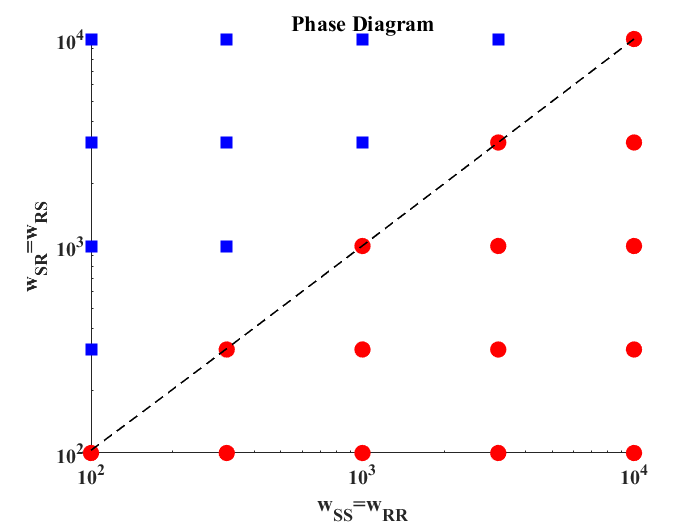}
    \caption{Phase diagram on stability of the base state ($S_{\rm hom}^*, I_{\rm hom}^*, R_{\rm hom}^*$) with random perturbations. Blue filled squares represent instability and red filled circles mean stability. The dashed line is plotted according to $w_{SR}w_{RS} > \left(w_{SS}+1/S_{\rm hom}^* \right)\left(w_{RR}+1/R_{\rm hom}^* \right)$, as given by the condition \reff{SIR DH lin.stb. condition Energy QC}. }
    \label{fig:Phase_Diagram}
\end{figure}
To further understand the instability induced by social interactions, we perform a series of simulations to study the influence of the parameters on the stability of the base state ($S_{\rm hom}^*,~ I_{\rm hom}^*, ~R_{\rm hom}^*$) with random perturbations. The simulations take $\lambda = 1$ for which the condition \reff{SIR DH lin.stb. condition Reaction} for reaction instability is not satisfied. Also, the parameters $w_{SS} = w_{RR}$ and $w_{RS} = w_{SR}$ are varied to explore the parameter space pertaining to stability. Figure~\ref{fig:Phase_Diagram} presents a phase diagram on the stability of the base state using various combinations of $w_{SS} = w_{RR}$ and $w_{RS} = w_{SR}$. The blue filled squares represent instability and red filled circles mean stability. In addition, a dashed hyperbolic curve is plotted as well according to $w_{SR}w_{RS} > \left(w_{SS}+1/S_{\rm hom}^* \right)\left(w_{RR}+1/R_{\rm hom}^* \right)$, which is the unstable condition \reff{SIR DH lin.stb. condition Energy QC} given by the linear stability analysis. The blue filled squares at upper left section indicate that relatively larger $w_{RS} = w_{SR}$, or less diagonal dominate of the interaction intensity matrix $W$, contributes the instability of the homogeneous base state, being quantitatively consistent with the unstable condition \reff{SIR DH lin.stb. condition Energy QC} given by the linear stability analysis. 
% Although there are some outliers for small parameters, the phase diagram gives stability results agreeing perfectly with the unstable condition \reff{SIR DH lin.stb. condition Energy QC} as the magnitude of the parameters grows. 
% As revealed by~\reff{SIR DH lin.stb. condition Energy QC}, there is a hyperbolic curve to 

% research the competition of the parameters $\rho$ and $\omega$. Figure \reff{fig:Phase_Diagram} shows the phase diagram with random perturbation, where filled red circle is stable and filled blue square is unstable. Dashed black line is the prediction of (\ref{SIR DH lin.stb. condition Energy QC}). 

% \subsection{Impact of instability on transmission}
\subsection{Cross interactions on transmission}
\begin{figure}[htbp]
    \centering
    \subfloat{
    \includegraphics[width=0.5\linewidth]{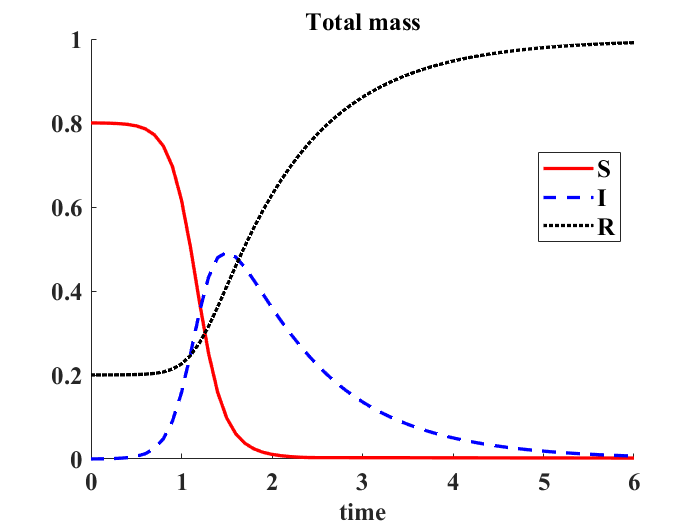}
    }
    \subfloat{
    \includegraphics[width=0.5\linewidth]{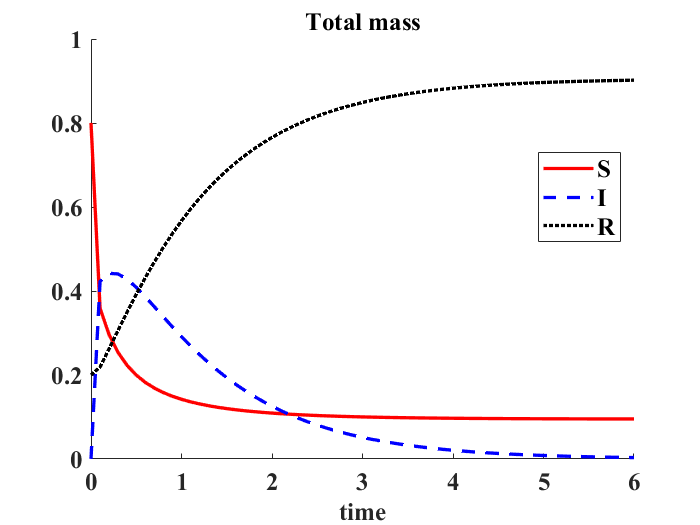}
    }
    \caption{Total mass of the susceptible, infectious and recovered individuals with high transmission rate $\lambda=10$. Left: $w_{SR} = w_{RS}=1$; Right: $w_{SR} = w_{RS}=5000$.}
    \label{fig:total mass}
\end{figure}
In this case, we study the impact of instability induced by social interactions on disease transmission with high transmission rate.  We first set $\lambda=10$, so that the condition \reff{SIR DH lin.stb. condition Reaction} holds. 
% The simulations start with the perturbed initial conditions~\reff{PertInit}.  
As disease outbreaks, the practice of social distancing, via social interactions, is often an effective measure to take to control the epidemic spreading. Here, we investigate the influence of social interactions on the evolution of total mass of different population in epidemic spreading. Figure~\ref{fig:total mass} displays the total mass of the susceptible, infectious and recovered individuals with various $w_{SR}$ (or $w_{RS}$). 
In the left plot, we take $w_{SR} = w_{RS}=1$ to understand the impact of normal intensity of social interactions. In the initial stage, total mass of susceptible, infectious, and recovered individuals remains constant for a while. After that, the disease outbreaks and the total mass of $I$ starts to grow quickly, with the total mass of susceptible individuals dropping to zero. In the very end, all individuals get infected and subsequently recovered. 
% In the upper left plot, we take $w_{SR} = w_{RS}=1$ to understand the impact of normal intensity of social interactions. In the initial stage, total mass of susceptible, infectious, and recovered individuals remains constant for a while. After that, the disease outbreaks and the total mass of $I$ starts to grow quickly, with the total mass of susceptible individuals dropping to zero. In the very end, all individuals get infected and recovered. Similar results can be found in the upper right plot that takes $w_{SR} = w_{RS}=100$, except that the peak value of $I$ decreases a bit. This indicates that the measure of social distancing can be practiced to delay the peak number of infected people.  To further understand the effect of the parameters, we also perform simulations with $w_{SR} = w_{RS}=500$. In such a case, the instability condition due to social interactions~\reff{SIR DH lin.stb. condition Energy QC} holds. It is of interest to find from spatial distribution that the $S$, $I$, and $R$ get separated immediately, so that $S$, $I$, and $R$ almost remain constant throughout the simulations, without causing much infection. Finally, as the parameter becomes $w_{SR} = w_{RS}=5000$,  we can find that 

In the right plot, we take $w_{SR} = w_{RS}=5000$ to understand the role of strong repulsive social interactions played in epidemic spreading. 
 In such a case, the instability condition due to social interactions~\reff{SIR DH lin.stb. condition Energy QC} holds. In contrast to the left plot, the epidemic spreads immediately with small perturbed initial conditions~\reff{PertInit} due to the instability induced by the strong repulsive social interactions, as indicated by the linear stability analysis in the Theorem \ref{t:LBA}. However, it is noteworthy to mention that the peak value of infected individuals is a bit lower than that for the case of $w_{SR} = w_{RS}=1$ shown in the left plot. After the outbreak, the total mass of infectious individuals decays quickly.   It is of interest to observe that with strong repulsive social interactions, a certain portion of susceptible individuals who accumulate at corners as in Figure~\ref{fig:E-type linear unstable} will not get infected throughout the epidemic. 
 % Moreover, our model predicts that the epidemic can end much sooner if a measure is taken to practice social distancing. 
 From such results, one can find that the practice of social distancing can effectively suppress the spreading of disease with lower infection peak number and more susceptible individuals who will not get infected throughout the epidemic.
 In summary, our development provides an effective tool for the administration to predict the evolution of epidemic with social interactions.   

% It is obvious that if the susceptible and recovered individuals keep the normal social distance, the epidemic will spread between people, even though the repulsion to the infectious individuals is strong. In the right plot, $\omega = 10^4$ and E-type instability is satisfied. With the isolation strategy between the susceptible and recovered individuals, the epidemic spreading is limited and some susceptible individuals will not be infected. Therefore in the high transmission rate case, the isolation between the susceptible and recovered individuals is helpful to control the epidemic spreading. 

\subsection{Isolation simulations}
\begin{figure}[htbp]
    \centering
    \subfloat{
    \includegraphics[width=0.33\linewidth]{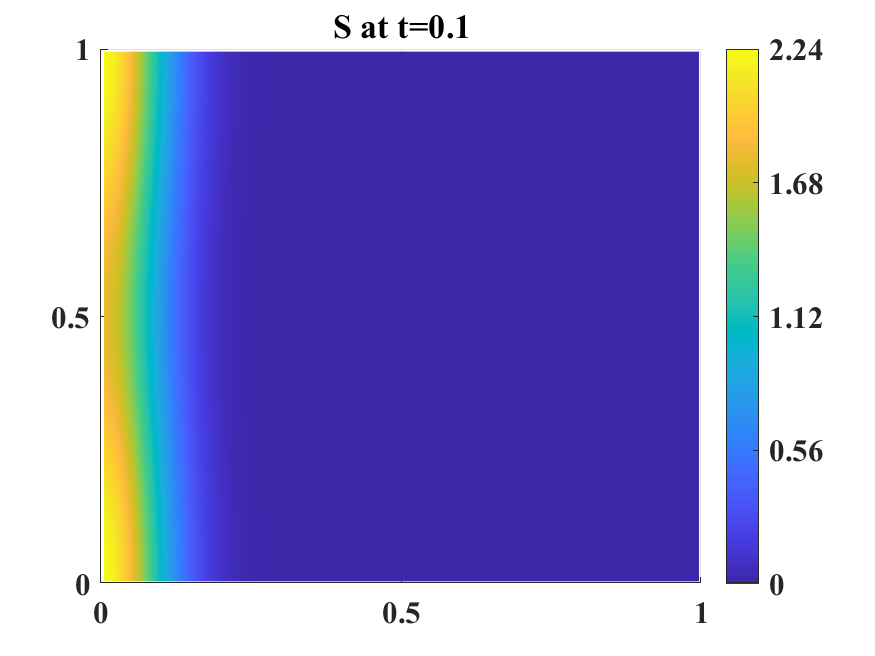}
    }
    \subfloat{
    \includegraphics[width=0.33\linewidth]{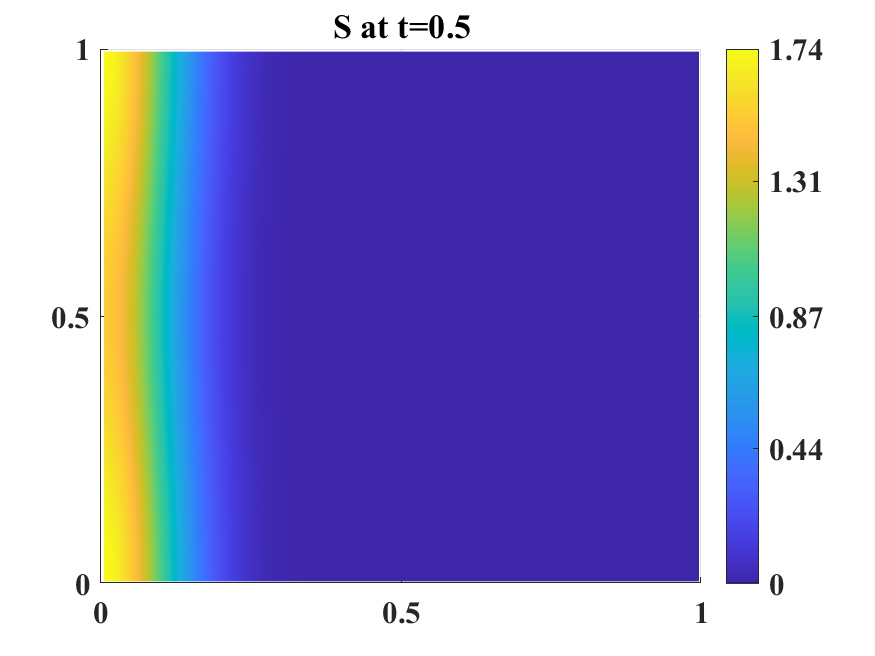}
    }
    \subfloat{
    \includegraphics[width=0.33\linewidth]{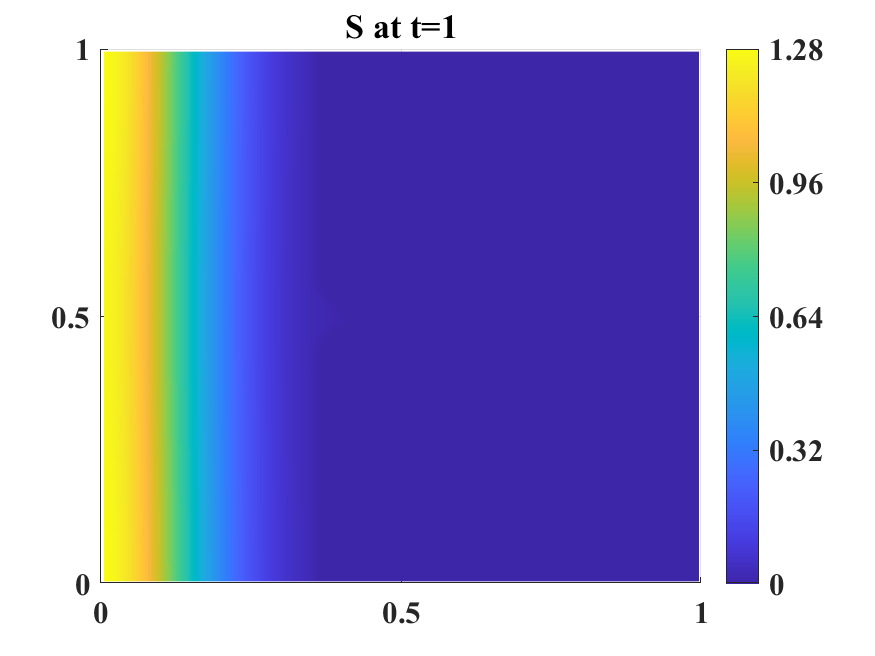}
    }

    \subfloat{
    \includegraphics[width=0.33\linewidth]{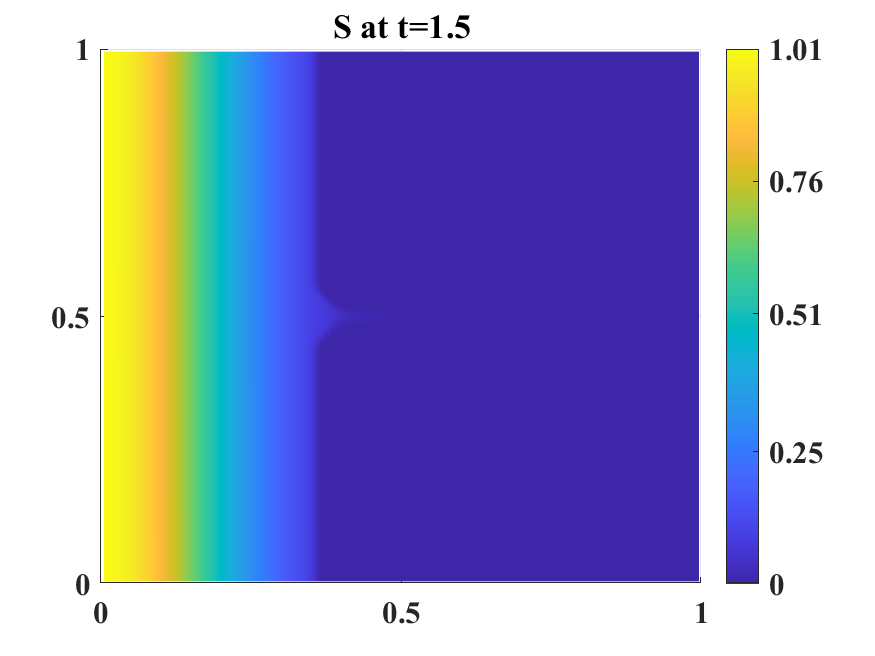}
    }
    \subfloat{
    \includegraphics[width=0.33\linewidth]{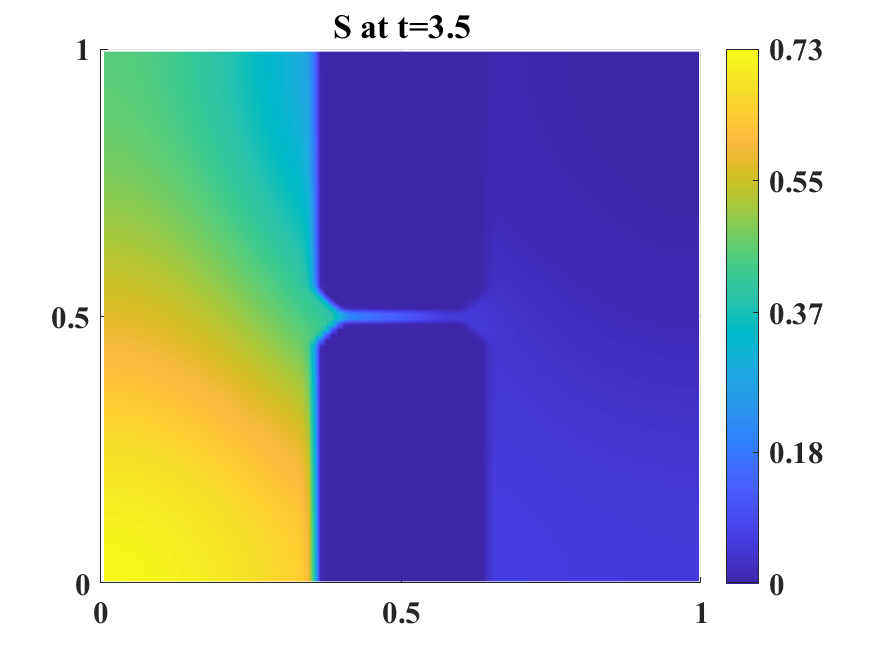}
    }
    \subfloat{
    \includegraphics[width=0.33\linewidth]{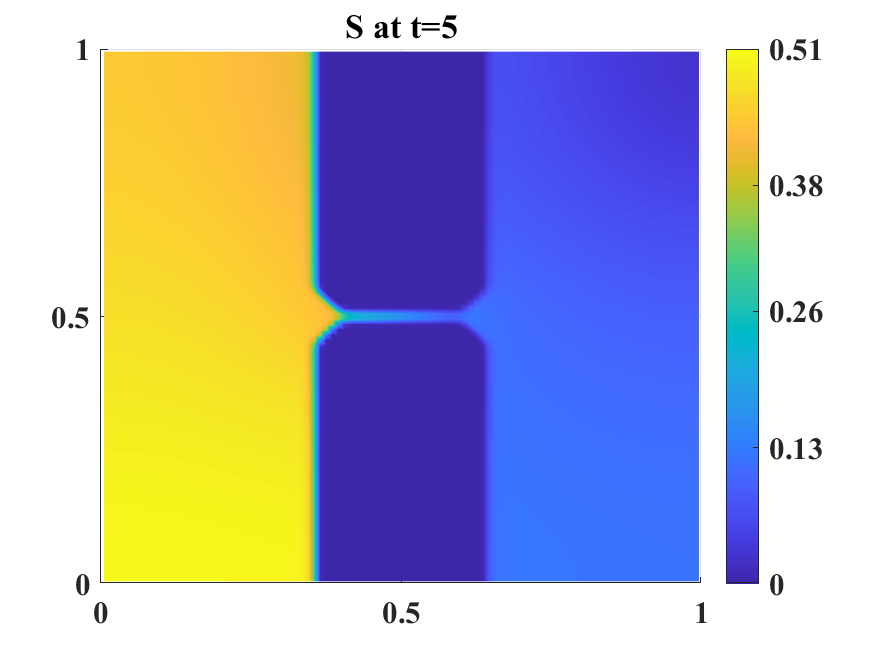}
    }
    \caption{Evolution of $S$ with the presence of an external potential in isolation simulations.}
    \label{fig:Qrt_S}
\end{figure}
\begin{figure}[htbp]
    \centering
    \subfloat{
    \includegraphics[width=0.33\linewidth]{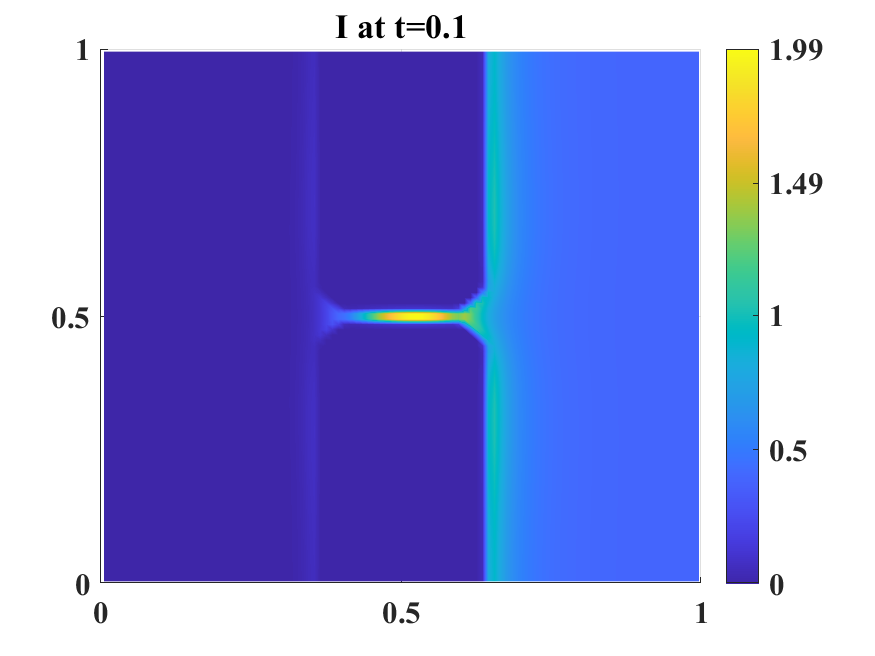}
    }
    \subfloat{
    \includegraphics[width=0.33\linewidth]{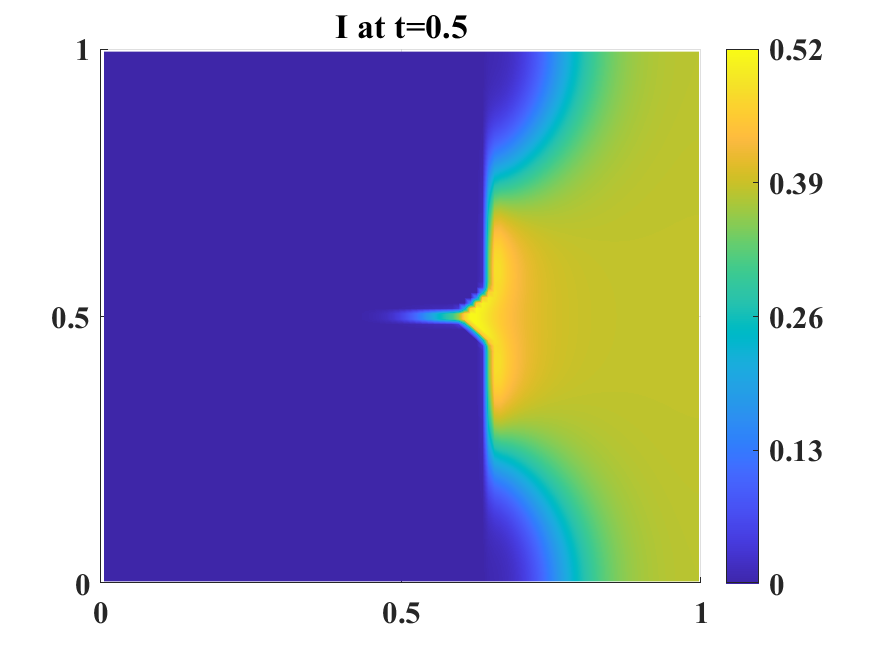}
    }
    \subfloat{
    \includegraphics[width=0.33\linewidth]{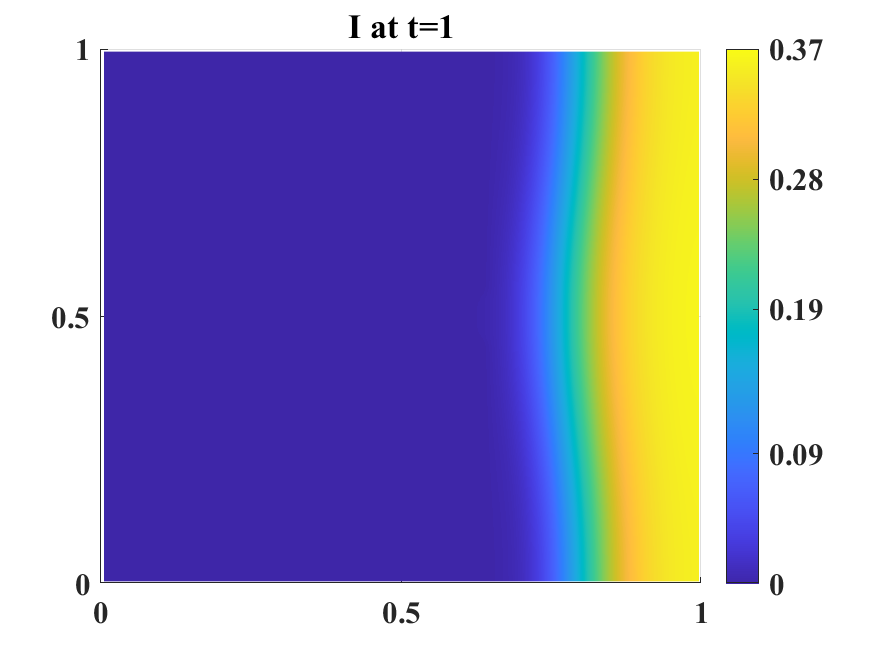}
    }

    \subfloat{
    \includegraphics[width=0.33\linewidth]{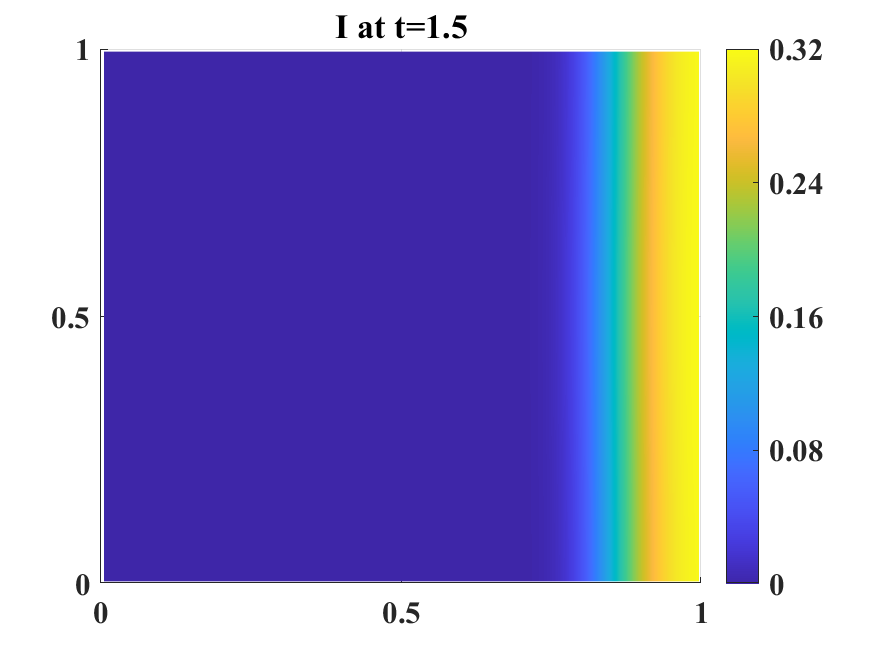}
    }
    \subfloat{
    \includegraphics[width=0.33\linewidth]{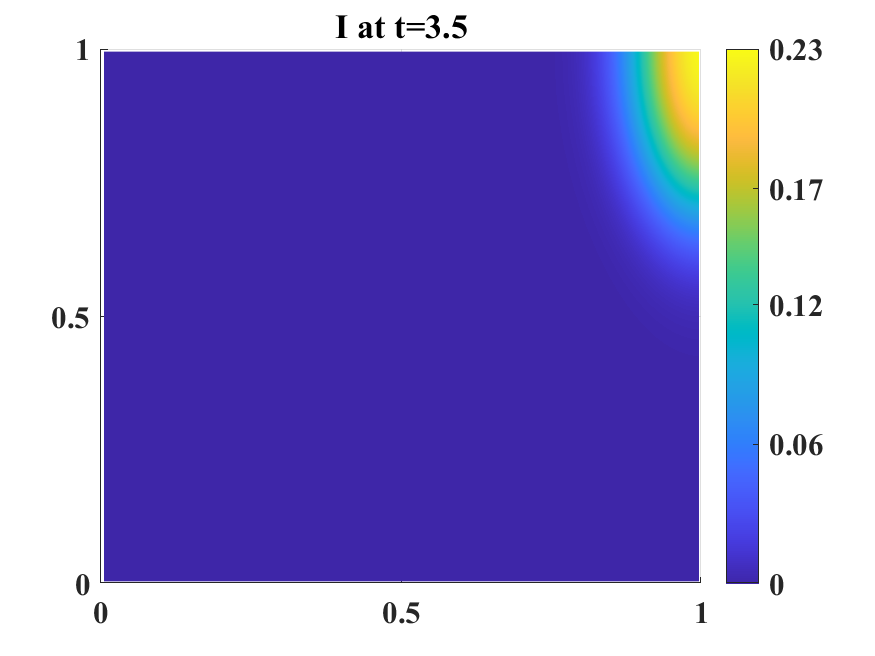}
    }
    \subfloat{
    \includegraphics[width=0.33\linewidth]{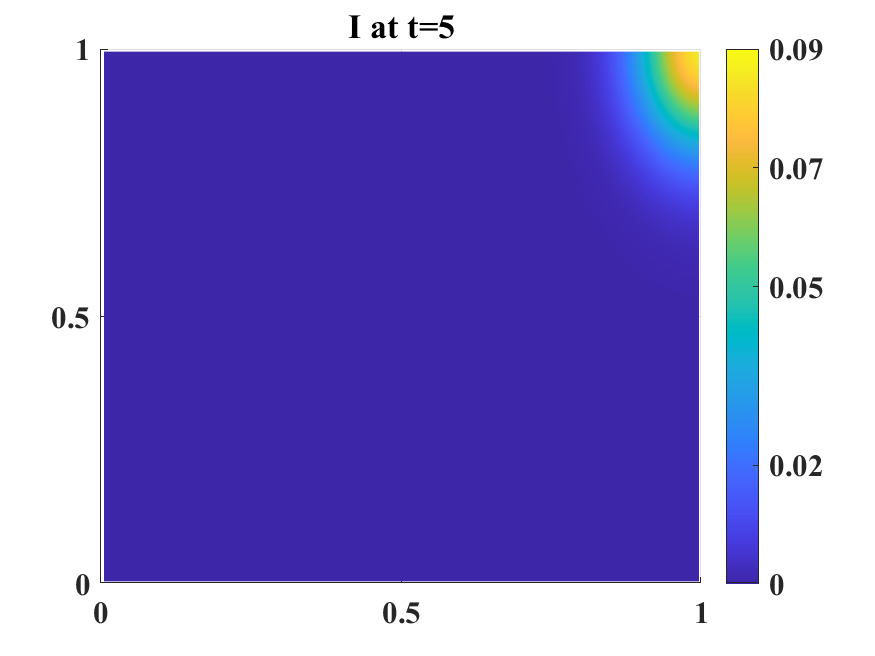}
    }
    \caption{Evolution of $I$ with the presence of an external potential in isolation simulations.}
    \label{fig:Qrt_I}
\end{figure}

\begin{figure}[htbp]
    \centering
    \subfloat{
    \includegraphics[width=0.33\linewidth]{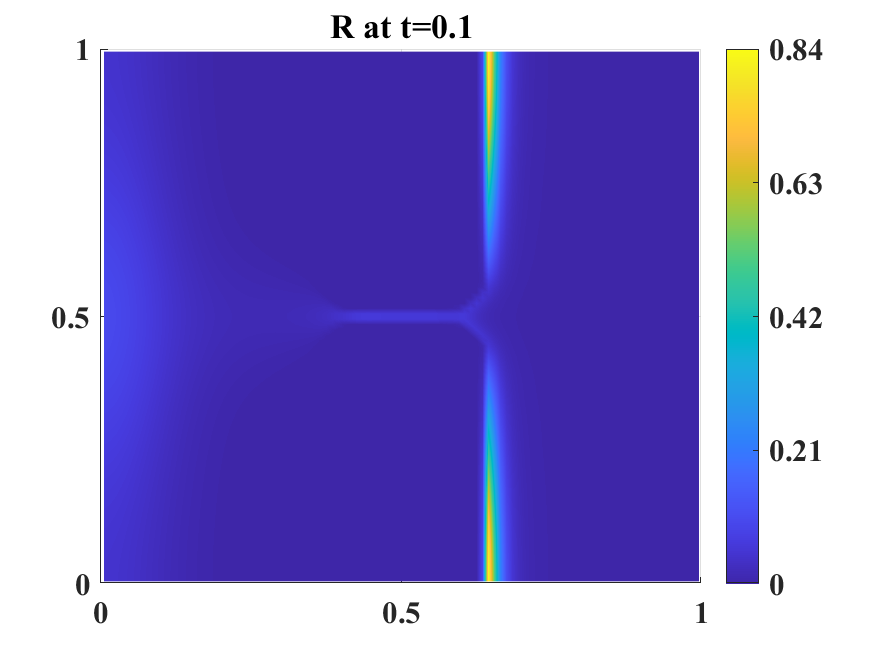}
    }
    \subfloat{
    \includegraphics[width=0.33\linewidth]{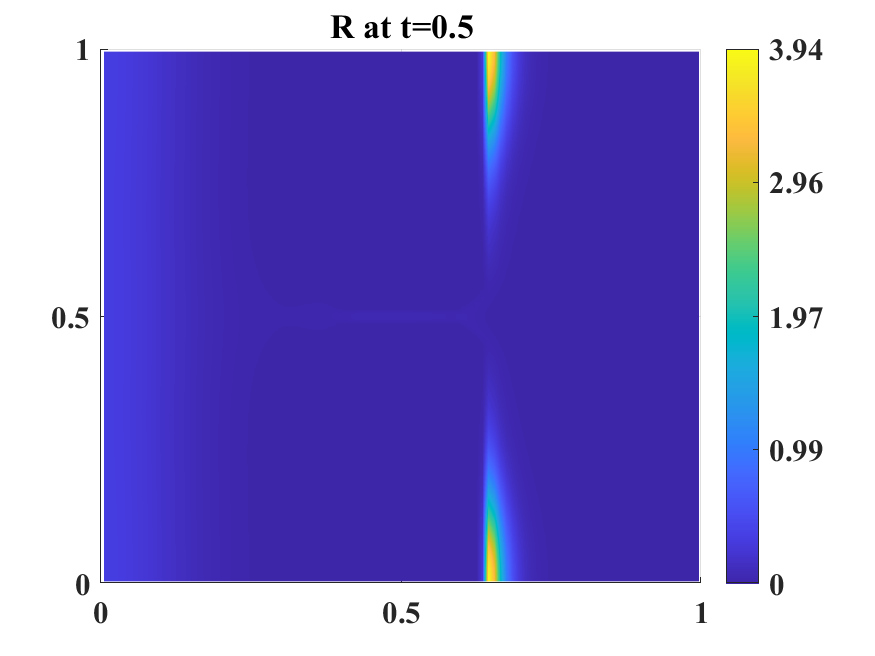}
    }
    \subfloat{
    \includegraphics[width=0.33\linewidth]{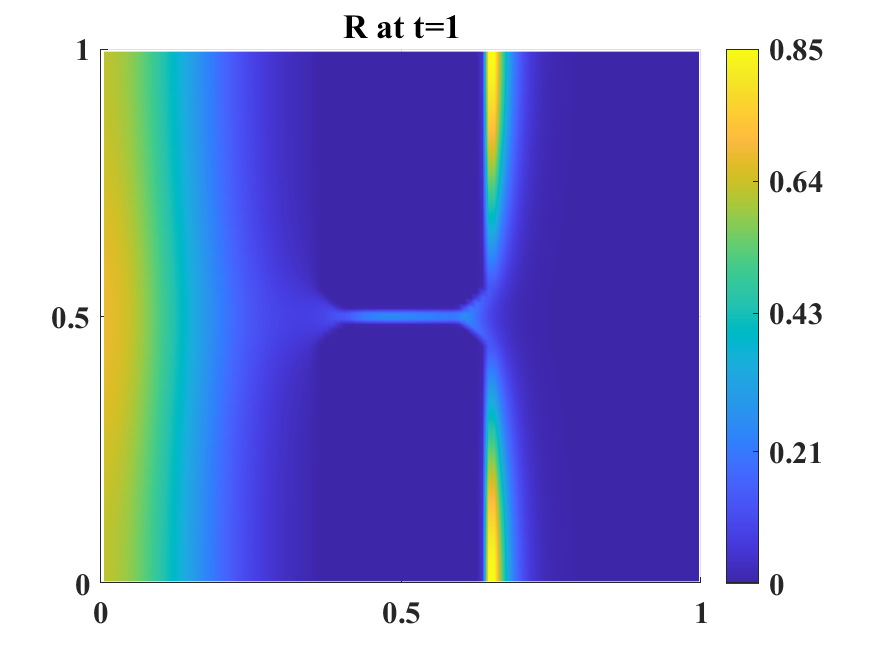}
    }

    \subfloat{
    \includegraphics[width=0.33\linewidth]{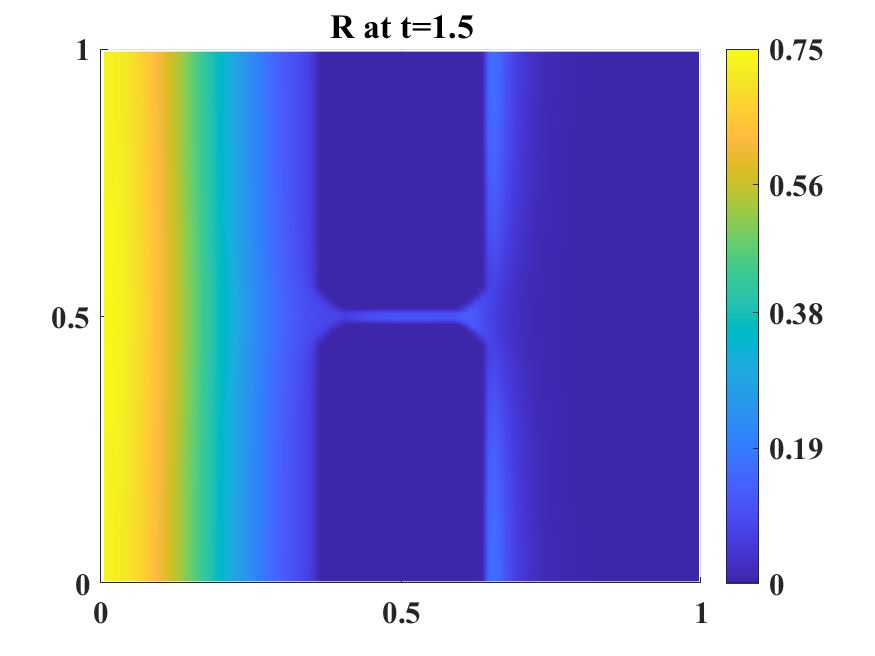}
    }
    \subfloat{
    \includegraphics[width=0.33\linewidth]{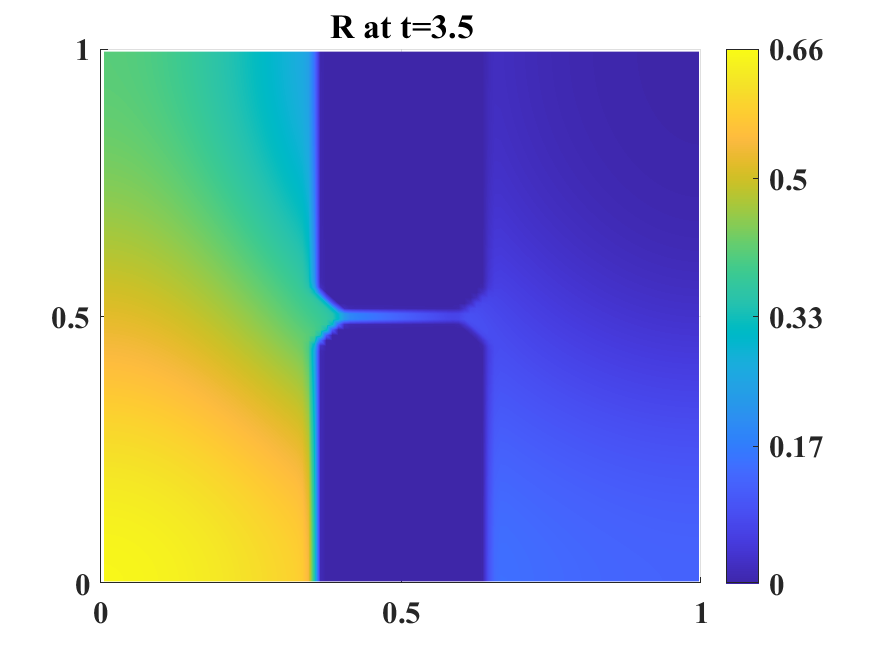}
    }
    \subfloat{
    \includegraphics[width=0.33\linewidth]{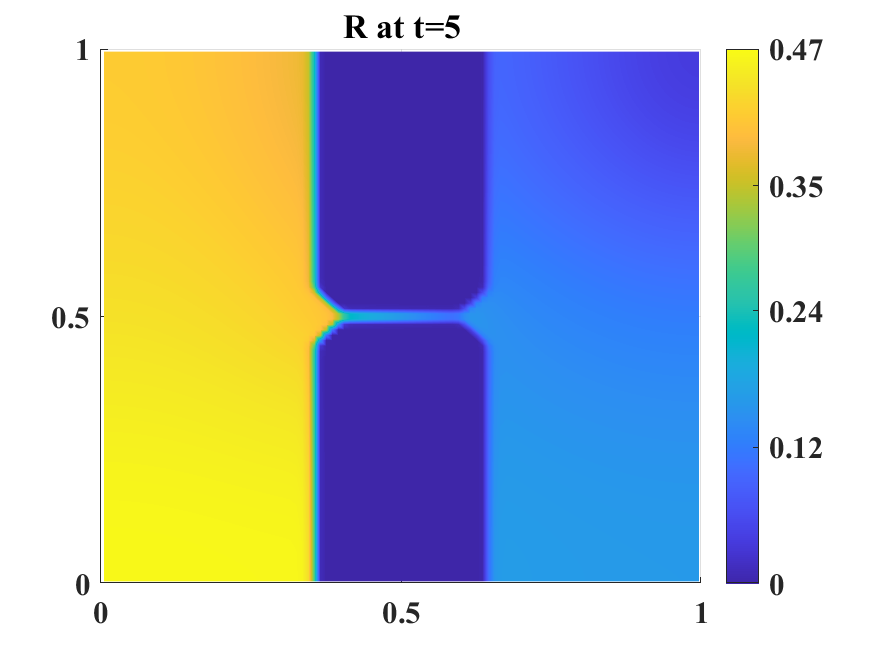}
    }
    \caption{Evolution of $R$ with the presence of an external potential in isolation simulations.}
    \label{fig:Qrt_R}
\end{figure}
The DFT-SIR-DH model~\reff{DH-SIR-DFT}
is applied to simulate the dynamics of susceptible, infectious, and recovered
individuals with the infected ones kept in isolation. An external potential is utilized to reflect the physical isolation:
\begin{equation} \nonumber
\begin{aligned}
    V_j^{\rm ext}(x,y) = -\left[\tanh(100y-55)+\tanh(-100y+45)+2\right]\cdot\left[\tanh(100x-60)+\tanh(-100x+40)\right]
\end{aligned}
\end{equation}
for $j=S, I, R$. Such an external potential divides the simulation domain into two compartments connected by an narrow passage, which is setup to mimic the essential communications.  The simulations start with an initial condition
\begin{equation} \nonumber
S_0(x,y) = \left\{
\begin{aligned}
    0.5,& \enspace x<0.4, \\
    0,& \enspace \text{otherwise},
\end{aligned}
\right. \enspace
I_0(x,y) = \left\{
\begin{aligned}
    0.5,& \enspace x>0.6, \\
    0,& \enspace \text{otherwise},
\end{aligned}
\right. \enspace
R_0(x,y) \equiv 0,
\end{equation}
which mimics a situation that the susceptible and infectious individuals are separated in the left and right compartments, respectively.  
The interaction intensity matrix is given by
\[
W = 
\begin{pmatrix}
1  & 1000 & 1 \\
1000  & 1000 & 1000\\
1 & 1000 & 1
\end{pmatrix}
,
\]
which prescribes that the social interactions between susceptible and recovered individuals are mild while the repulsive interactions between infected individuals and others are remarkably strong.

We present the evolution of susceptible, infected and recovered individuals at various time snapshots in Figure~\ref{fig:Qrt_S}, \ref{fig:Qrt_I}, and \ref{fig:Qrt_R}, respectively. As shown in Figure~\ref{fig:Qrt_S}, the susceptible individuals first quickly escape away from the narrow channel passage connecting  to infectious individuals and concentrate at left boundary of the domain $\Omega$. Due to diffusion as well as the redistribution of infectious individuals, the susceptible individuals gradually move back and distribute evenly in the left region. As seen from Figure~\ref{fig:Qrt_I}, infection first peaks in the narrow passage, as expected. Due to the strong repulsion,   the infected individuals gradually relocate to the right compartment, and eventually fade away as clusters at the northeastern corner. 
% After At the time $t=3.5$, the infected individuals gather as a cluster. 
Figure~\ref{fig:Qrt_R} illustrates that the recovered individuals who first get infected in the narrow passage move to the right compartment and relocate at the thin layer next to the physical isolation. After that, due to the strong repulsion between the infected ones, the recovered individuals redistribute through the narrow channel and gradually move to the left compartment. In the long run, the recovered individuals mix with susceptible individuals and  distribute evenly in the left compartment. From above results, one can observe that by tuning the interaction intensity matrix, our DFT-SIR-DH model with the local interaction Debye-H\"uckel potential can effectively capture the social interactions between different individuals.

\section{Conclusion} \label{Sec. conclusion}
Social interactions play a pivotal role in epidemic spreading. It is highly desirable to incorporate such interactions in theoretical modeling. This work has proposed a local mean-field density functional theory model by using the sum-of-exponential approximation of convolution kernels for social interactions. Such approximation converts the convolution terms into interaction potentials that are governed by the Debye-H\"uckel equation. Thanks to the local formulation of the proposed model, linear stability analysis has been able to derive a novel instability condition associated with cross interactions. Global existence of the solution to the proposed model with a simplified self-repulsive interaction potential has been established as well. Extensive numerical simulations have been performed to study the impact of social interactions on epidemic spreading, verify the instability conditions obtained from linear stability analysis, and provide theoretical guides for the control of disease spreading. 

% derive the SIRDH model
% In this work, we use mean-field approximation to investigate the effect of social distancing and apply it to epidemic spreading. Combining susceptible-infected-recovered (SIR) model, dynamical density functional theory (DDFT) and mass-action transmission mechanism together, we propose the SIR Debye-H\"uckel model (\ref{SIR DH model}).  
% % linear stability analysis
% Linear stability analysis indicates that the instability of the homogeneous steady state results from the transmission terms and free energy. Therefore we classify the the instability conditions (\ref{SIR DH lin.stb. condition Reaction}) and (\ref{SIR DH lin.stb. condition Energy}) into two distinct types, R-type and E-type, respectively. More precisely, with given homogeneous steady state, increasing $\lambda$ or decreasing $\gamma$ may cause the so-called R-type instability. Meanwhile, E-type instability occurs when $w_{SR}$ increases or $\kappa$ decreases. 
% % pde analysis
% In PDE analysis, we prove the $L^p$ boundedness of $S(\bm{x},t),I(\bm{x},t)$ and $R(\bm{x},t)$. With the help of the technique of Moser-Alikakos iteration, we show the global existence of the model. 
% % numerical results
% In numerical simulation, we find patterns with different parameters. Through numerical results, we show that the R-type instability (\ref{SIR DH lin.stb. condition Reaction}) contributes to the change in densities while the E-type instability (\ref{SIR DH lin.stb. condition Energy}) results in spatially inhomogeneous distribution. 

% acknowledgements
\vskip 5mm
\noindent{\bf Acknowledgements.}
This work was supported by National Key R\&D Program of China 2023YFF1204200. The authors thank Professor Youshan Tao  for helpful discussions.

% bibliography
\bibliographystyle{unsrt}
\bibliography{SIR_DDFT}

\end{document}